\documentclass[3p,times]{elsarticle}

\usepackage{ecrc}

\volume{00}

\firstpage{1}

\journalname{Procedia Computer Science}

\runauth{}

\jid{procs}

\jnltitlelogo{Procedia Computer Science}

\CopyrightLine{2011}{Published by Elsevier Ltd.}

\usepackage{amssymb}
\usepackage[figuresright]{rotating}

\usepackage{afterpage}
\usepackage{makecell}
\usepackage{tabulary}
\usepackage{xcolor}
\usepackage{color}
\usepackage{colortbl}
\usepackage{prettyref}
\usepackage{framed}
\usepackage{threeparttable}
\usepackage{booktabs}       % professional-quality tables
\usepackage[utf8]{inputenc} % allow utf-8 input
\usepackage[T1]{fontenc}    % use 8-bit T1 fonts
\usepackage{hyperref}       % hyperlinks
\usepackage{url}            % simple URL typesetting
\usepackage{booktabs}       % professional-quality tables
\usepackage{amsfonts}       % blackboard math symbols
\usepackage{nicefrac}       % compact symbols for 1/2, etc.
\usepackage{microtype}      % microtypography
\usepackage{xcolor}         % colors
\usepackage{amsthm}
\usepackage{amsmath}
\usepackage{amssymb}
\usepackage{algorithm}
\usepackage{algpseudocode}
\usepackage{graphicx}
\usepackage{subfigure}
\usepackage{booktabs} 
\usepackage{bm}
\usepackage{hyperref}
\usepackage{cleveref}
\usepackage{enumitem}
\usepackage{mathtools}
\usepackage{longtable}
\usepackage{graphicx}
\usepackage{natbib}
\usepackage{array}
\usepackage{lipsum}
\usepackage{multirow}
\usepackage{makecell}
\usepackage{xcolor}
\usepackage{tikz}
\usepackage{etoolbox}
\newcommand{\pref}[1]{\prettyref{#1}}

\newcommand{\savehyperref}[2]{\texorpdfstring{\hyperref[#1]{#2}}{#2}}
\newrefformat{eq}{\savehyperref{#1}{Eq. \textup{(\ref*{#1})}}}
\newrefformat{eqn}{\savehyperref{#1}{Eq.~\textup{(\ref*{#1})}}}
\newrefformat{lem}{\savehyperref{#1}{Lemma~\ref*{#1}}}
\newrefformat{assump}{\savehyperref{#1}{Assumption~\ref*{#1}}}
\newrefformat{defi}{\savehyperref{#1}{Definition~\ref*{#1}}}\newrefformat{tab}{\savehyperref{#1}{Table~\ref*{#1}}}
\newrefformat{lemma}{\savehyperref{#1}{Lemma~\ref*{#1}}}
\newrefformat{line}{\savehyperref{#1}{Line~\ref*{#1}}}
\newrefformat{thm}{\savehyperref{#1}{\textbf{Theorem}~\ref*{#1}}}
\newrefformat{corr}{\savehyperref{#1}{Corollary~\ref*{#1}}}
\newrefformat{cor}{\savehyperref{#1}{Corollary~\ref*{#1}}}
\newrefformat{sec}{\savehyperref{#1}{Section~\ref*{#1}}}
\newrefformat{app}{\savehyperref{#1}{\textbf{Appendix}~\ref*{#1}}}
\newrefformat{appendix}{\savehyperref{#1}{\textbf{\ref*{#1}}}}
\newrefformat{ex}{\savehyperref{#1}{Example~\ref*{#1}}}
\newrefformat{fig}{\savehyperref{#1}{Figure~\ref*{#1}}}
\newrefformat{alg}{\savehyperref{#1}{Algorithm~\ref*{#1}}}
\newrefformat{rem}{\savehyperref{#1}{Remark~\ref*{#1}}}
\newrefformat{conj}{\savehyperref{#1}{Conjecture~\ref*{#1}}}
\newrefformat{prop}{\savehyperref{#1}{Proposition~\ref*{#1}}}
\newrefformat{proto}{\savehyperref{#1}{Protocol~\ref*{#1}}}
\newrefformat{prob}{\savehyperref{#1}{Problem~\ref*{#1}}}
\newrefformat{claim}{\savehyperref{#1}{Claim~\ref*{#1}}}
\newrefformat{que}{\savehyperref{#1}{Question~\ref*{#1}}}
\newrefformat{op}{\savehyperref{#1}{Open Problem~\ref*{#1}}}
\newrefformat{fn}{\savehyperref{#1}{Footnote~\ref*{#1}}}
\newrefformat{property}{\savehyperref{#1}{Property~\ref*{#1}}}

\usepackage{xspace}
\usepackage{amsmath}
\usepackage{bbm}
\usepackage{mathrsfs}
\usepackage{bm}
\usepackage{makecell}
\usepackage{tabulary}
\newcommand{\calRO}{\mathcal{O}}

\newcommand{\calD}{\mathcal{D}}

\newcommand{\distpara}{W^{\calD}}
\newcommand{\truepara}{W^{\calRO}}

\newcommand{\E}{\mathbb{E}}

\newcommand{\wtilde}{\widetilde}

\newtheorem{assumption}{Assumption}[section]

\newcommand*{\circled}[1]{\lower.7ex\hbox{\tikz\draw (0pt, 0pt)%
    circle (.5em) node {\makebox[1em][c]{\small #1}};}}
\robustify{\circled}

\newtheorem{thm}{Theorem}[section]

\newtheorem{lem}[thm]{Lemma}

\usepackage{afterpage}
\usepackage{prettyref}
\usepackage{pgfplots}
\usepackage{caption}
\usepackage{amsmath,amssymb,amsfonts}
\usepgfplotslibrary{fillbetween}
\usepackage{framed}
\newrefformat{eq}{\savehyperref{#1}{Eq. \textup{(\ref*{#1})}}}
\newrefformat{eqn}{\savehyperref{#1}{Eq.~\textup{(\ref*{#1})}}}
\newrefformat{lem}{\savehyperref{#1}{\textbf{Lemma}~\ref*{#1}}}
\newrefformat{lemma}{\savehyperref{#1}{\textbf{Lemma}~\ref*{#1}}}
\newrefformat{defi}{\savehyperref{#1}{\textbf{Definition}~\ref*{#1}}}
\newrefformat{line}{\savehyperref{#1}{Line~\ref*{#1}}}
\newrefformat{thm}{\savehyperref{#1}{\textbf{Theorem}~\ref*{#1}}}
\newrefformat{corr}{\savehyperref{#1}{Corollary~\ref*{#1}}}
\newrefformat{cor}{\savehyperref{#1}{Corollary~\ref*{#1}}}
\newrefformat{sec}{\savehyperref{#1}{Section~\ref*{#1}}}
\newrefformat{table}{\savehyperref{#1}{Table~\ref*{#1}}}
\newrefformat{app}{\savehyperref{#1}{\textbf{Appendix}~\ref*{#1}}}
\newrefformat{assum}{\savehyperref{#1}{Assumption~\ref*{#1}}}
\newrefformat{ex}{\savehyperref{#1}{Example~\ref*{#1}}}
\newrefformat{fig}{\savehyperref{#1}{Figure~\ref*{#1}}}
\newrefformat{alg}{\savehyperref{#1}{Algorithm~\ref*{#1}}}
\newrefformat{rem}{\savehyperref{#1}{Remark~\ref*{#1}}}
\newrefformat{conj}{\savehyperref{#1}{Conjecture~\ref*{#1}}}
\newrefformat{prop}{\savehyperref{#1}{Proposition~\ref*{#1}}}
\newrefformat{proto}{\savehyperref{#1}{Protocol~\ref*{#1}}}
\newrefformat{prob}{\savehyperref{#1}{Problem~\ref*{#1}}}
\newrefformat{claim}{\savehyperref{#1}{Claim~\ref*{#1}}}
\newrefformat{que}{\savehyperref{#1}{Question~\ref*{#1}}}
\newrefformat{op}{\savehyperref{#1}{Open Problem~\ref*{#1}}}
\newrefformat{fn}{\savehyperref{#1}{Footnote~\ref*{#1}}}

\usepackage{xspace}
\usepackage{amsmath}
\usepackage{bbm}
\usepackage{mathrsfs}
\usepackage{amssymb}
\usepackage{bm}
\usepackage{makecell}
\usepackage{tabulary}

\newtheorem{defi}{Definition}[section]

\begin{document}

\begin{frontmatter}



% thanks命令重定义成星号
% \renewcommand{\thefootnote}{\fnsymbol{footnote}} % 使脚注符号变为符号而非数字
% \renewcommand*{\thanks}[1]{\footnotemark[\value{footnote}]\protected@xdef\@thanks{\@thanks
%     \protect\footnotetext[\value{footnote}]{#1}}} % 重定义thanks命令
\renewcommand{\thefootnote}{\fnsymbol{footnote}}

\dochead{}
%% e.g. \dochead{17th International Conference on Dynamical Processes in Excited States of Solids}

\title{Deciphering the Interplay between Attack and Protection Complexity in Privacy-Preserving Federated Learning}

%use optional labels to link authors explicitly to addresses:
% PCL Research Center of Networks and Communications

%Version 1
\author[label1]{Xiaojin Zhang
\footnote{School of Computer Science and Technology, Huazhong University of Science and Technology, Wuhan, 430074, China. email: xiaojinzhang@hust.edu.cn}}
\address[label1]{Huazhong University of Science and Technology, China}
\author[label1]{Mingcong Xu}
\author[label1]{Yiming Li}
\author[label1]{Wei Chen\footnote{email: lemuria\_chen@hust.edu.cn}}
\author[label5]{Qiang Yang\footnote{Corresponding Author, email:qyang@cse.ust.hk}}
\address[label5]{The Hong Kong Polytechnic University, China}

%\address[label6]{Hong Kong University of Science and Technology, China}

\begin{abstract}
      Federated learning (FL) offers a promising paradigm for collaborative model training while preserving data privacy. However, its susceptibility to gradient inversion attacks poses a significant challenge, necessitating robust privacy protection mechanisms. This paper introduces a novel theoretical framework to decipher the intricate interplay between attack and protection complexities in privacy-preserving FL. We formally define "Attack Complexity" as the minimum computational and data resources an adversary requires to reconstruct private data below a given error threshold, and "Protection Complexity" as the expected distortion introduced by privacy mechanisms. Leveraging Maximum Bayesian Privacy (MBP), we derive tight theoretical bounds for protection complexity, demonstrating its scaling with model dimensionality and privacy budget. Furthermore, we establish comprehensive bounds for attack complexity, revealing its dependence on privacy leakage, gradient distortion, model dimension, and the chosen privacy level. Our findings quantitatively illuminate the fundamental trade-offs between privacy guarantees, system utility, and the effort required for both attacking and defending. This framework provides critical insights for designing more secure and efficient federated learning systems.
\end{abstract}


\end{frontmatter}

% \linenumbers

% \section{}
% \label{}

%\maketitle
\section{Introduction}
Federated learning (FL) has emerged as a pivotal distributed machine learning paradigm, enabling multiple clients to collaboratively train a shared global model without directly sharing their raw local data \cite{mcmahan2017communication}. This inherent characteristic positions FL as a strong candidate for applications where data privacy and ownership are paramount, such as healthcare, finance, and mobile computing. By keeping sensitive data localized, FL aims to mitigate privacy risks associated with centralized data aggregation.

Despite its privacy-preserving promise, FL is not immune to sophisticated privacy attacks. Recent research has demonstrated that sensitive information about clients' private training data can be inferred from the shared model updates, particularly gradients, through techniques known as gradient inversion attacks \cite{zhu2019deep}. These attacks, such as Deep Leakage from Gradients (DLG) and its variants, can reconstruct individual data samples with alarming accuracy, thereby undermining the fundamental privacy guarantees of FL. Consequently, a critical challenge in FL research lies in developing robust privacy protection mechanisms that effectively counter these threats without unduly compromising model utility or system efficiency.

Existing privacy-preserving techniques in FL, including differential privacy (DP) \cite{dwork2014algorithmic}, homomorphic encryption (HE) \cite{gentry2009fully}, and secure multi-party computation (SMC) \cite{mugunthan2019smpai}, aim to obfuscate or encrypt shared information. While effective to varying degrees, these methods often introduce trade-offs: stronger privacy typically comes at the cost of reduced model accuracy, increased computational overhead, or slower convergence \cite{zhang2022no, zhang2024meta, zhang2024game, zhang2023probably, zhang2023towards, zhang2023theoretically}. A comprehensive understanding of these trade-offs, particularly from the perspective of an attacker's effort versus a defender's cost, remains largely unexplored. There is a pressing need for a rigorous theoretical framework that quantifies the difficulty of both attacking and protecting private data in FL, thereby providing actionable insights for system design.

This paper addresses this gap by introducing and formally defining two novel metrics: "Attack Complexity" and "Protection Complexity." Attack complexity quantifies the minimum resources (e.g., attacker iterations, data samples) an adversary requires to reconstruct private data within a specified error tolerance. Conversely, protection complexity measures the amount of distortion introduced by a privacy mechanism to achieve a certain privacy guarantee. By analyzing the interplay between these two complexities, we aim to provide a deeper understanding of the fundamental limits and trade-offs in privacy-preserving federated learning.

We formally define "Attack Complexity" as the minimum number of attacker iterations required to achieve a predefined data reconstruction error threshold, providing a quantitative measure of privacy protection from an adversarial perspective. We introduce "Protection Complexity" to characterize the expected distortion applied to model parameters by privacy mechanisms, offering a direct measure of the privacy cost. Our main contributions are summarized as follows:
\begin{itemize}
    \item We derive tight theoretical bounds for Protection Complexity under the Maximum Bayesian Privacy (MBP) framework. Unlike prior work that primarily relies on Local Differential Privacy (LDP), our analysis directly quantifies distortion in terms of MBP, providing a more general characterization of privacy-utility trade-offs.
    \item We establish comprehensive upper and lower bounds for Attack Complexity, demonstrating its intricate dependency on factors such as privacy leakage, gradient distortion, model dimensionality, and the chosen privacy budget (MBP). These bounds quantitatively illustrate how stronger privacy guarantees elevate the effort required for successful attacks.
    \item We provide a theoretical foundation for understanding the intricate trade-offs between privacy, utility, and the computational resources expended by both attackers and defenders in federated learning environments.
\end{itemize}
The remainder of this paper is organized as follows: Section \ref{sec:related_work} reviews related work. Section \ref{sec:preliminary} introduces the federated learning setting, threat model, and privacy protection mechanisms. Section \ref{sec:defining_complexity} formally defines attack and protection complexity. Section \ref{sec:protection_bounds} and \ref{sec:attack_bounds} present the theoretical bounds for protection and attack complexities, respectively. Finally, Section \ref{sec:conclusion} concludes the paper and outlines future research directions.

\section{Related Work}
\label{sec:related_work}
The landscape of federated learning (FL) research has rapidly expanded, driven by the dual objectives of collaborative intelligence and data privacy. This section reviews key areas related to our work, including the evolution of FL, the emergence of gradient inversion attacks, and existing privacy-preserving mechanisms, highlighting how our proposed framework of attack and protection complexity offers a novel perspective.

\subsection{Federated Learning and Privacy Concerns}
Federated learning, pioneered by Google \cite{mcmahan2017communication}, enables decentralized model training by aggregating local model updates rather than raw data. Early FL algorithms like FedSGD and FedAvg \cite{mcmahan2017communication} demonstrated the feasibility of this approach. While FL inherently offers a degree of privacy by keeping data local, it quickly became apparent that sharing model updates, particularly gradients, could still leak sensitive information. This realization spurred extensive research into privacy-preserving FL. The core challenge lies in balancing model utility with strong privacy guarantees, a trade-off that often dictates the practical applicability of FL systems.

\subsection{Gradient Inversion Attacks}
The vulnerability of FL to gradient leakage was starkly demonstrated by gradient inversion attacks. Deep Leakage from Gradients (DLG) \cite{zhu2019deep} showed that an adversary, given only the shared gradients, could reconstruct the original training data samples, including images and text. Subsequent works, such as improved Deep Leakage from Gradients (iDLG) \cite{zhao2020idlg} and InvertGrad \cite{geiping2020inverting}, refined these techniques, making reconstruction more efficient and accurate, even for batch training. These attacks highlight the critical need for robust defense mechanisms and underscore the importance of quantifying the effort required for such adversarial actions. Our definition of "Attack Complexity" directly addresses this need by formalizing the resources an attacker must expend to achieve a desired reconstruction quality.

\subsection{Privacy Protection Mechanisms in Federated Learning}
To counter gradient inversion attacks and other privacy threats, various protection mechanisms have been proposed. Differential Privacy (DP) \cite{dwork2014algorithmic} is a widely adopted technique that adds carefully calibrated noise to gradients or model parameters before sharing them. This ensures that the presence or absence of any single data record in the training set does not significantly alter the output distribution, thereby providing strong privacy guarantees. Homomorphic Encryption (HE) \cite{gentry2009fully} allows computations on encrypted data, enabling secure aggregation without revealing individual updates. Secure Multi-Party Computation (SMC) \cite{mugunthan2019smpai} enables multiple parties to jointly compute a function over their private inputs without revealing those inputs to each other. While effective, these methods introduce overheads: DP can degrade model accuracy, and HE/SMC incur significant computational and communication costs. Our concept of "Protection Complexity" quantifies the distortion introduced by such mechanisms, providing a metric to assess their inherent cost in terms of utility or information loss. Our work distinguishes itself by introducing a unified framework that explicitly links the \textit{effort} required for an attack to the \textit{cost} of protection. Specifically, our definitions of "Attack Complexity" and "Protection Complexity" provide a novel lens through which to analyze these trade-offs.

\section{Preliminary}\label{sec:preliminary}
We approach the problem from the attacker's perspective, introducing a mathematical framework to analyze the difficulty of reconstructing private data in federated learning. To quantify this difficulty, we define a new metric termed attack complexity, which measures the computational and data resources required for an attacker to achieve a given reconstruction error threshold \( \tau \). Through theoretical analysis, we derive an upper bound on attack complexity, allowing us to investigate its relationship with key variables such as the reconstruction error threshold \( \tau \), the number of attacker iterations \( T \), and the size of the local dataset \( |D^{(k)}| \).

This section aims to present the necessary symbols and background for understanding the concepts discussed later in the paper. It introduces the federated setup, the types of attackers considered, and the privacy protection mechanisms employed in this work. Table \ref{tab:notation} summarizes the symbols used throughout this paper.

\begin{longtable}{cc}
\caption{Notation} \label{tab:notation} \\
\toprule
\textbf{Symbol} & \textbf{Description} \\
\midrule
\( S_k(\tau) \) & Attack complexity of client \( k \) under error threshold \( \tau \) \\
\( \tau \) & Predefined error threshold \\
\( d \) & Error function (distance metric), such as MSE, SSIM, PSNR \\
\( T \) & Total number of attacker iterations \\
\( \tilde{X}^{(k)}_{t,i} \) & The \( i \)-th data sample reconstructed by the attacker for client \( k \) in iteration \( t \) \\
\( X^{(k)}_{i} \) & The \( i \)-th real data sample in client \( k \) \\
\( D^{(k)} \) & Local dataset of client \( k \) \\
\( \epsilon_p^{(k)} \) & Privacy leakage degree of client \( k \) \\
\( \truepara \) & Original gradient  \\
\( \distpara \) & Perturbed gradient  \\
\( \Delta^{(k)} \) & Gradient distortion degree of client \( k \) \\
\( \gamma^{(k)} \) & Confidence parameter of client \( k \) \\
\bottomrule
\end{longtable}

\subsection{Federated Learning Setting}

The objective of this study is to investigate horizontal federated learning, where multiple participants collaboratively train a global model while maintaining the privacy of their local datasets. In horizontal federated learning, each participant's local dataset is denoted as \( D^{(k)} \), and the optimization goal of the federated learning process is to:

\[
\min_{\theta} \sum_{k=1}^{K} \frac{m_k}{m} \mathcal{L}_k(\theta)
\]

where \( \theta \) represents the global model parameters, \( \mathcal{L}_k(\theta) \) is the loss function for the \( k \)-th participant, \( m_k \) is the size of the local dataset for the \( k \)-th participant, and \( m \) is the total size of all datasets. The objective is to minimize the weighted loss across all participants' local models, ensuring both privacy preservation and model performance.

Several algorithms have been proposed to solve this optimization problem, including FedSGD, FedAvg, and FedProx. The theoretical foundation of this paper primarily builds upon the FedSGD algorithm, whose main process is as follows:

\begin{enumerate}
    \item Each client computes the gradient of the loss function with respect to its local model parameters.
    \item The client sends the computed gradients to the central server.
    \item The server aggregates the received gradients, typically by averaging, and updates the global model parameters.
    \item The updated global model is sent back to the clients for further local training.
\end{enumerate}

While FedSGD is a classic and widely used algorithm, previous studies have shown that it may lead to data leakage from the clients. Specifically, during the training process, an attacker could exploit the parameters exchanged between the client and the server to infer the local data of individual clients. In the following sections, we will introduce the attack strategies considered in this study, as well as the defense mechanisms implemented to mitigate these risks.

\subsection{Threat Model}

As mentioned earlier, attackers may exploit the parameters exchanged during federated learning to recover client data. In this study, we treat the server as the attacker, whose goal is to utilize the parameter information sent by client \( k \) to accurately recover any client’s private data. We assume that the server is a semi-honest attacker, meaning it will follow the federated learning protocol faithfully for parameter exchange and updates -- this is the "honest" part. However, it may attempt to investigate the clients' information to recover private data.

More specifically, the attacker can exploit the gradients sent by clients to infer sensitive information about their local datasets. Existing methods for achieving such attacks include Deep Leakage Gradient (DLG), improved Deep Leakage Gradient (iDLG), and InvertGrad, among others. In this paper, the attacker uses the DLG method, with the optimization goal being:

\[
\min_{\bar{x}, \bar{y}} \left\| \nabla \mathcal{L}_k (x,y) - \nabla \mathcal{L}_k (\bar{x}, \bar{y}) \right\|_2
\]

where \( \nabla \mathcal{L}_k (x,y) \) represents the gradient of the loss function for client \( k \) with respect to the true local data \( x,y \), and $\nabla \mathcal{L}_k (\bar{x}, \bar{y})$ represents the estimated model parameters based on the information inferred. The attacker aims to minimize the distance between the actual gradients and the gradients obtained from the server's inferred data, thereby reconstructing the client's local data.

\subsection{Privacy Protection Mechanism via Parameter Distortion}
As mentioned earlier, the server can exploit the parameters transmitted by clients to reconstruct their private data. To mitigate this risk, previous studies have proposed various protection techniques, such as perturbation and encryption, to secure the transmitted parameters. A classic perturbation method involves adding Differential Privacy noise to the gradient information, while encryption techniques include homomorphic encryption and secure multi-party computation. Collectively, these methods are referred to as parameter distortion. In this study, the privacy protection method considered is gradient perturbation.

Under these protection mechanisms, the federated training process (e.g., using FedSGD) is modified as follows: the client receives the parameters sent by the server, performs local training to compute updated gradients, and then perturbs the gradients (e.g., by adding noise) before uploading them. This protects against potential gradient inversion attacks by the server. The server, in turn, utilizes the perturbed gradients to update the global model parameters and sends the updated parameters back to the clients, repeating the process iteratively.

It is important to note that while gradient perturbation protects data privacy, the introduced perturbations may degrade model performance or reduce efficiency. Therefore, balancing the magnitude of the perturbations to achieve a trade-off between utility and privacy remains a critical challenge.

\section{Defining Attack Complexity and Protection Complexity}\label{sec:defining_complexity}
Attack complexity is described as the number of iterations \( T \) required by the attacker to ensure that the distance between the reconstructed data and the real data \( d \leq \tau \), where \( \tau \) is the predefined error threshold. In this context, attack complexity assists in assessing the degree of data privacy protection. Qualitatively, the larger the attack complexity, the better the privacy protection.

The quantitative analysis of attack complexity can be conducted as follows:

\begin{enumerate}
    \item \textbf{Define Distance Metric:}
    
    Choose an appropriate \textbf{distance metric} to measure the distance \( d \) between the reconstructed data and the real data. For instance, in image reconstruction tasks, metrics such as Mean Squared Error (MSE), Structural Similarity Index (SSIM), or Peak Signal-to-Noise Ratio (PSNR) can be used.
    
    \item \textbf{Define Attack Complexity:}
\end{enumerate}

\begin{defi} [Attack Complexity]  
\label{definition1}
Define the error function \( d \colon \mathbb{R}^m \times \mathbb{R}^m \to \mathbb{R}^+ \), representing the distance between the reconstructed data and the real data (e.g., MSE). Under the error threshold \( \tau \), the attack complexity is defined as:

\[
S_{k}(\tau) = \min \left\{ T \mid \mathbb{E}_{D^{(k)}} \left[ \frac{1}{|D^{(k)}|} \sum_{i=1}^{|D^{(k)}|} \frac{1}{T} \sum_{t=1}^{T} d \left( \tilde{X}^{(k)}_{t,i}, X^{(k)}_{i} \right) \right] \leq \tau \right\}.
\]

Where:
\begin{itemize}
    \item \( T \) is the total number of attack iterations.
    \item \( \tilde{X}^{(k)}_{t,i} \) represents the \( i \)-th reconstructed data sample for client \( k \) generated by the attacker in iteration \( t \).
    \item \( d(\cdot) \) is the error metric function used to measure the discrepancy between \( \tilde{X}^{(k)}_{t,i} \) and \( X^{(k)}_{i} \). The appropriate metric (e.g., MSE, SSIM, PSNR) should be chosen based on the task.
    \item \( \tau \in \mathbb{R}^+ \) is the preset error threshold, indicating the maximum allowable average error.
    \item \( \mathbb{E}_{D^{(k)}} \) denotes the expectation over the randomness of the local dataset \( D^{(k)} \).
\end{itemize}
\end{defi}

The attack complexity definition establishes a quantitative relationship between data reconstruction error \( d \) and the strength of privacy protection, from the perspectives of attack iterations. Specifically, \( S_{k}(\tau) \) characterizes the minimum number of iterations \( T \) required for the attacker to achieve a reconstruction error below the given threshold \( \tau \), reflecting how privacy protection constrains the attack complexity. Through this definition, the trade-off between privacy protection, client data size, attacker computational resources, and reconstruction error is clearly revealed. This relationship provides theoretical guidance for designing privacy-preserving mechanisms, particularly in optimizing attacker costs.

Next, we define the notion of protection complexity, which characterizes the level of perturbation applied to original model parameters to achieve privacy protection.

\begin{defi} [Protection Complexity]  \label{definition:protection-complexity}
Let \( \truepara \in \mathbb{R}^m \) denote the original (unperturbed) model parameter vector of a client, and let \( \distpara \in \mathbb{R}^m \) denote the corresponding perturbed vector after applying a protection mechanism. Then the protection complexity is defined as the expected distortion introduced by the mechanism:
\[
C = \mathbb{E} \left[ \| \distpara - \truepara \|_2^2 \right].
\]
\begin{itemize}
    \item \( \truepara \) is the original model parameter or gradient.
    \item \( \distpara \) is the perturbed version released to the server.
    \item \( \| \cdot \|_2^2 \) denotes the squared Euclidean distance.
    \item The expectation is taken over the randomness of the protection mechanism.
\end{itemize}
\end{defi}

This definition reflects the amount of information distortion induced by the protection mechanism. A higher protection complexity generally corresponds to stronger privacy guarantees but may come at the cost of reduced utility or degraded convergence performance.

\section{Protection Complexity Bounds for Protection Mechanisms with Parameter Distortion}\label{sec:protection_bounds}
% Protection complexity is measured via distortion extent. In this section, we provide bounds for this distortion extent using Maximum Bayesian Privacy (MBP).

{In this work, we adopt MBP as the foundational privacy notion for analyzing protection complexity. This decision is guided by several critical considerations.}

{First, MBP provides a natural measure of privacy for inference attacks, such as DLG, by quantifying how much an adversary’s belief about the private data is updated after observing the released outputs~\cite{zhang2022no}. Compared to traditional privacy notions such as LDP, which measure output distribution differences over worst-case neighboring datasets, MBP instead captures the discrepancy between the prior and posterior distributions which makes MBP more suitable for settings like federated learning, where attackers often exploit background knowledge and model parameter to perform inference.}

{Second, MBP admits estimation algorithms that approximate posterior updates from observed outputs, and recent work has provided formal analyses of the associated estimation errors~\cite{zhang2023meta}. In contrast, the privacy analysis of LDP relies on worst-case sensitivity estimation and  LDP lacks a general mechanism for analyzing the estimation error of its privacy guarantees. These limitations make MBP a more analytically flexible and practically applicable framework for evaluating privacy in federated learning settings.}

\paragraph{The Adversary}
We assume that the adversary knows the conditional distribution used by the attacker. That is, the adversary knows the distribution $f_{\distpara|\truepara}(w^d|w^o)$. Equipped with his/her prior belief on the private data, the posterior belief distribution of the attacker is then approximated.

\paragraph{Reconstruction}
Upon observing $\distpara$, the adversary reconstructs $\truepara$. Given the reconstructed model parameter, the private data could be reconstructed using the optimization algorithm such as DLG.

First, we revisit and incorporate several key definitions and results proposed in \cite{zhang2024deciphering}, which provide a theoretical foundation for connecting LDP with MBP.

\subsection{Theoretical Relationship between MBP and LDP}

\begin{defi}[$\epsilon$-Local Differential Privacy, adapted from \cite{zhang2024deciphering}]
A randomized function $\mathcal{M}$ is said to be $\epsilon$-locally differentially private if, for all possible datasets $d$ and $d'$, and for all subsets $S$ of the output space, the following condition holds:
\begin{equation}
    \frac{\Pr[\mathcal{M}(d) \in S]}{\Pr[\mathcal{M}(d') \in S]} \leq e^{\epsilon}.
\end{equation}
Here, $\epsilon \geq 0$ quantifies the level of privacy protection, and $\Pr[\mathcal{M}(d) \in S]$ denotes the probability that the randomized mechanism $\mathcal{M}$ outputs an element in the set $S$ when applied to dataset $d$.
\end{defi}

This definition ensures that the output distribution of the randomized mechanism remains statistically similar across any two individual inputs, thereby protecting user-level data contributions against inference attacks.

To model privacy from a probabilistic inference perspective, \cite{zhang2024deciphering} further introduced the notion of MBP, which measures the adversary's ability to update their belief about the data upon observing system outputs.

\begin{defi}[$\xi$-Maximum Bayesian Privacy, adapted from \cite{zhang2024deciphering}]\label{def:mbp}
Let $F$ denote the belief distribution held by the attacker regarding the private data, and let $f$ be the associated probability density function. A system is said to satisfy $\xi$-Maximum Bayesian Privacy if, for any $w \in \mathcal{W}$ and $d \in \mathcal{D}$, the following inequality holds:
\begin{equation}
    e^{-\xi} \leq \frac{f_{D|W}(d \mid w)}{f_D(d)} \leq e^{\xi}.
\end{equation}
\end{defi}

This criterion quantifies the disparity between the posterior and marginal distributions of the data. A smaller $\xi$ implies that an adversary's observation has minimal influence on the belief update, indicating stronger privacy guarantees.

Importantly, \cite{jiang2018context} established a generalized connection between LDP and MBP, showing that these two privacy notions can be quantitatively related under certain conditions on the prior knowledge of the adversary.

\begin{lem}[Theoretical Relationship between MBP and LDP, adapted from theorem 1 of ~\cite{jiang2018context}]\label{Relationship between MBP and LDP}
If an algorithm satisfies $\xi$-LDP, then it also satisfies $\xi $-MBP. Conversely, if an algorithm satisfies $\xi$-MBP, then it satisfies $2\xi$-LDP.
\end{lem}

This lemma bridges the gap between privacy measured in output similarity (LDP) and privacy measured in belief robustness (MBP), and it provides useful guidance for converting guarantees between different privacy notions under mild assumptions.

\subsection{Optimal Bound for Protection Complexity}

First, we derive the lower bound for protection complexity with $\epsilon$-Maximum Bayesian Privacy. Please refer to \pref{sec: lower_bound_for_federated_utility_loss} for the full analysis.

To establish the lower bound rigorously, we first introduce the key technical assumption and a foundational lemma, which follow the approach of \cite{bhowmick2018protection}.

\begin{assumption}[Compositional Differential Privacy Bound, adapted from~\cite{bhowmick2018protection}]\label{assump:eps-DP}
For each $i$ and $t$, and for all $W^{\cal{D}}_{i,t}$, the channel mapping $X_i$ to $W^{\cal{D}}_{i,t}$ is $\epsilon_{i,t}(W^{\cal{D}}_{i,t})$-approximately differentially private. There exists a constant $\varepsilon_{\mathrm{kl}}$ such that
\begin{equation}
\sum_{i=1}^n \sum_{t=1}^T \mathbb{E} \left[
\min\left\{
\frac{\left(\epsilon_{i,t}^{2}(W^{\cal{D}}_{i,t})\right)}{\log 2},
\epsilon_{i,t}(W^{\cal{D}}_{i,t})
\right\}
\right]
\leq n \cdot \varepsilon_{\mathrm{kl}},
\end{equation}
where $\varepsilon_{\mathrm{kl}}$ represents the average per-user KL-privacy leakage over all users and time steps. The expectation is taken over the randomness in the privatized outputs $W^{\cal{D}}_{i,t}$.
\end{assumption}

We also leverage a standard result from information-theoretic lower bound analysis:

\begin{lem}[Assouad’s Lemma, adapted from~\cite{bhowmick2018protection}]\label{lem:assouad}
Let $\mathcal{V}$ be a class of binary vectors that induce $\delta$-separation in the Hamming metric, meaning that any two elements in $\mathcal{V}$ differ in at least $\delta$ coordinates. Then the minimax risk satisfies
\[
\mathfrak{M}_n(\truepara(\mathcal{P}), L, Q) \geq \frac{\delta}{2} \sum_{j=1}^m \left(1 - \left\| \mathbb{P}_{+j} - \mathbb{P}_{-j} \right\|_{\mathrm{TV}} \right),
\]
where $\mathbb{P}_{+j}$ and $\mathbb{P}_{-j}$ denote the output distributions corresponding to $V_j = +1$ and $V_j = -1$, respectively.
\end{lem}

With Lemma~\ref{Relationship between MBP and LDP}, Assumption~\ref{assump:eps-DP}, and Lemma~\ref{lem:assouad}, we are ready to present the lower bound theorem.

\begin{thm}[Lower Bound for Protection Complexity]\label{thm:lower-bound-eps-DP}
Assume that the mechanism $M : \mathbb{S}^{m-1} \rightarrow \mathbb{S}^{m-1}$ is $\epsilon$-MBP for input $x \in \mathbb{S}^{m-1}$, where $\epsilon \leq m$. Then let $\truepara$ represent the original model information, and $\distpara$ represent the distorted model information. Assume that both the original and perturbed model parameters lie on the unit sphere, that is, $\truepara, \distpara \in \mathbb{S}^{m-1}$, we have:
\[
\mathbb{E}[\|\distpara - \truepara\|_2^2] \geq c \cdot \max\left( \frac{m}{\min(\epsilon, \epsilon^2)}, 1 \right).
\]
\end{thm}

Theorem~\ref{thm:lower-bound-eps-DP} establishes a fundamental lower bound on the expected distortion between the true model parameters and the perturbed outputs, under the MBP framework. This result reveals a tight trade-off between utility and privacy: in order to guarantee a given level of Bayesian privacy \(\epsilon\), any mechanism must introduce distortion that scales at least as \(\left(\frac{m}{\min(\epsilon, \epsilon^2)}\right)\), where \(m\) is the model dimensionality. 

{The lower bound in Theorem~\ref{thm:lower-bound-eps-DP} shares a similar scaling behavior with Proposition 3 in~\cite{bhowmick2018protection}, both showing that the protection complexity scales with $\frac{m}{\min(\epsilon, \epsilon^2)}$. However, the two results are derived under fundamentally different privacy frameworks. The result in~\cite{bhowmick2018protection} is derived under DP or its relaxed variants such as \((\epsilon, \delta)\)-DP and Rényi DP, and is primarily used to characterize how privacy constraints affect learning performance. In contrast, our analysis is based on MBP, and aims to quantify the inherent distortion introduced by the protection mechanism itself.}

With Lemma~\ref{Relationship between MBP and LDP}, we now present the optimal bound theorem.

\begin{thm}[Optimal Bound for Protection Complexity]\label{thm: upper_bound_utility_loss}
Let $m$ represent the dimension of the vector, and $\epsilon$ represent the level of MBP. Let $\truepara$ the original model information, and $\distpara$ represent the distorted model information. Assume that both the original and perturbed model parameters lie on the unit sphere, that is, $\truepara, \distpara \in \mathbb{S}^{m-1}$, combining Lemma~\ref{Relationship between MBP and LDP}, we have:
\begin{align}
    \E[\|\distpara - \truepara\|_2^2]= \Theta (\frac{m}{\min\{\epsilon^2, \epsilon\}}).
\end{align}
\end{thm}

Theorem~\ref{thm: upper_bound_utility_loss} establishes a tight characterization of protection complexity under the $\epsilon$-MBP framework. It shows that the expected distortion introduced by any mechanism satisfying a given MBP level must scale as $\Theta\left( \frac{m}{\min\{\epsilon^2, \epsilon\}} \right)$, where $m$ denotes the dimension of the parameter vector. This result quantitatively reveals how the degree of perturbation depends jointly on the privacy budget $\epsilon$ and the dimensionality of the protected data. In particular, smaller $\epsilon$ or higher-dimensional models necessitate larger distortion, highlighting the fundamental trade-off between privacy and fidelity in mechanism design.

{This result differs from the upper bound in Proposition 4 of~\cite{bhowmick2018protection}, which analyzes the output \(Z\) of a specific mechanism (PrivUnit2) applied to a unit vector \(u \in \mathbb{S}^{m-1}\), and establishes a mechanism-specific bound on \(\mathbb{E}[\|Z - u\|_2^2]\). In contrast, our theorem applies more broadly to all mechanisms satisfying MBP, and focuses on the distortion between true model parameters \(\truepara\) and their protected counterparts \(\distpara\), which more directly captures privacy-preserving model parameters in federated learning scenarios.}

% add mbp estimation
% begin
{
\color{black}
It is important to note that the bound established in Theorem \ref{thm: upper_bound_utility_loss} depends on the MBP parameter $\epsilon$, which is often unknown in practical applications and must be empirically estimated. Any such estimation process inevitably introduces an error.

In ~\pref{sec: anlaysis_for_estimation_mbp}, we detail and analyze an algorithm for estimating $\epsilon$. The key result is that the absolute error between the estimated value $\hat\epsilon$ and the true value $\epsilon$ is bounded and smaller as the number of simulation rounds $T_{sim}$ increases.

\begin{thm}[Optimal Bound for Protection Complexity Under The Estimation of $\epsilon$]\label{thm: the_estimation_error_of_epsilon}
Let $m$ represent the dimension of the vector, $\hat{\epsilon}$ represent the estimation level of MBP, and $T_{sim}$ represent the total number of simulation rounds. Let $\truepara$ the original model information, and $\distpara$ represent the distorted model information. Assume that both the original and perturbed model parameters lie on the unit sphere, that is, $\truepara, \distpara \in \mathbb{S}^{m-1}$. With at least $1 - \delta$, there exists a constant $c$, we have:
\begin{align}
    \E[\|\distpara - \truepara\|_2^2] = \Omega \left(\frac{m}{\min\{\left(\hat\epsilon\ + \zeta\right)^2, \hat\epsilon+ \zeta\}}\right),
\end{align}
and
\begin{align}
    \E[\|\distpara - \truepara\|_2^2] = O \left(\frac{m}{\min\{\left(\hat\epsilon- \zeta\right)^2, \hat\epsilon- \zeta\}}\right).
\end{align}
where $\zeta = c \cdot \sqrt{\frac{\ln(2/\delta)}{T_{sim}}}$.
\end{thm}

This result implies that for a sufficiently large $T_{sim}$, we can obtain a highly reliable estimate of the true level $\epsilon$. Consequently, the protection complexity calculated using the estimate $\hat\epsilon$ will be a precise approximation of the theoretical value characterized by Theorem~\ref{thm: upper_bound_utility_loss}. This validates the practical applicability of our theoretical framework, providing a sound basis for designing and calibrating privacy-preserving mechanisms.
}
% add mbp estimation
% end

\section{Attack Complexity Bounds for Protection Mechanisms with Parameter Distortion}\label{sec:attack_bounds}
The goal of this section is to establish the bound of attack complexity for gradient inversion attacks. To achieve this, we first quantify the privacy leakage during the attack process.

In gradient inversion attacks, the attacker aims to reconstruct the private dataset \( D^{(k)} \) of a client using gradient $\theta^{(k)}$ shared during training. The client computes its true gradient:
\[
W^{\calRO(k)} = \frac{1}{|D^{(k)}|} \sum_{i=1}^{|D^{(k)}|} \nabla L(\theta, X_i^{(k)}, Y_i^{(k)}),
\]
and applies a protection mechanism. The attacker uses algorithms like gradient matching to reconstruct the dataset, denoted as \( D_{\text{attacker}}^{(k)} \), by minimizing the discrepancy between the reconstructed and true data over \( T \) iterations, we represent this discrepancy as:
\[
\frac{1}{|D_{\text{attacker}}^{(k)}|} \sum_{i=1}^{|D_{\text{attacker}}^{(k)}|} \frac{1}{T} \sum_{t=1}^T \| \tilde{X}_{t,i}^{(k)} - X_i^{(k)} \|.
\]

\begin{defi}[Privacy Leakage]
We define privacy leakage based on the discrepancy between the reconstructed data and the original data. Here, we assume that the attacker targets all local data, i.e., \( |D^{(k)}| = |D^{(k)}_{\text{attacker}}| \).

\[
    \epsilon_p^{(k)} = 1 - \mathbb{E}\left[\frac{1}{|D^{(k)}|}\sum_{i=1}^{|D^{(k)}|}\frac{1}{T}\sum_{t=1}^T \frac{\|\tilde{X}_{t,i}^{(k)} - X_{i}^{(k)}\|}{\tau}\right],
\]
\end{defi}

where \( D^{(k)} \) is the local dataset of client \( k \), \( T \) is the total number of attack iterations, \( \tilde{X}_{t,i}^{(k)} \) represents the reconstructed data at iteration \( t \), \( X_i^{(k)} \) is the original data, \( \tau \) is the error threshold, and \( \mathbb{E} \) denotes the expectation. A smaller \( \epsilon_p^{(k)} \) indicates stronger privacy protection.

For gradient inversion attacks, to protect data privacy, the client perturbs the generated original gradients, producing perturbed gradients which are then uploaded to the server. Therefore, we formally define the degree of this perturbation.

\begin{defi}[Gradient Distortion Degree] We introduce gradient distortion to mimic the client's protection by adding perturbations to the gradients:
\[
    \Delta^{(k)} = \| \truepara - \distpara \|_2
\]
\end{defi}

where \( \truepara \) denotes the original gradient, and \( \distpara \) represents the perturbed gradient. Larger \( \Delta^{(k)} \) reflects stronger privacy protection but may also impact the model’s convergence and performance.

Combining the above definitions, we present an upper bound on attack complexity with at least $1 - \gamma^{(k)}$ confidence:

\textcolor{black}{\begin{thm}[Upper Bound of Attack Complexity with Gradient Expectation Bound]\label{upper-bound}
Let us define the variables as follows: \( c_a \), \( c_b \), and \( c_2 \) are constants; \( D^{(k)} \) represents the local dataset of client \( k \); \( \tau \) denotes the error threshold; \( T \) is the number of attacker iterations; \( m \) is the dimension of the model gradient vector; and \( \epsilon \) denotes the level of MBP enforced by the protection mechanism. Combining the definitions of privacy leakage and gradient distortion, we provide the upper bounds for \( T \) with at least \( 1 - \gamma^{(k)} \) confidence as follows:
\begin{align}
    T^{\frac{1}{2}} \leq \frac{c_2 \cdot c_b}{
        c_a \cdot \Delta^{(k)}
        - \tau \cdot \left[1 - \epsilon_p^{(k)} + \sqrt{\frac{\ln(2/\gamma^{(k)})}{2|D^{(k)}|}} \right]
    }
\end{align}
By invoking the protection complexity result from Section~\ref{sec:protection_bounds}, where Theorems~\ref{thm:lower-bound-eps-DP} and~\ref{thm: upper_bound_utility_loss} jointly establish a tight distortion bound $\mathbb{E}[\|\distpara - \truepara\|_2^2] = \Theta \left( \frac{m}{\min\{\epsilon^2, \epsilon\}} \right)$, we estimate the gradient distortion as $\Delta^{(k)} = \sqrt{ \frac{m}{\min\{\epsilon, \epsilon^2\}} }$, which reflects the expected magnitude of perturbation introduced by the protection mechanism. Substituting this into the inequality gives:
\begin{align}
    T^{\frac{1}{2}} \leq \frac{c_2 \cdot c_b}{c_a \cdot \left( \sqrt{\frac{m}{\min\{\epsilon^2, \epsilon\}}- \sqrt{\frac{m\cdot ln(1/\gamma^{(k)})}{2}}} \right) - \tau \cdot \left[1 - \epsilon_p^{(k)} + \sqrt{\frac{\ln(2/\gamma^{(k)})}{2|D^{(k)}|}} \right]}
\end{align}
\end{thm}}

\textcolor{black}{The inequality establishes the relationship between the number of attacker iterations \( T \), the privacy leakage \( \epsilon_p^{(k)} \), and the gradient distortion, which now depends on the model dimension \( m \) and the privacy budget \( \epsilon \). When the attack complexity $T$ is held constant, increasing the local dataset size $|D^{(k)}|$ leads to a reduction in privacy leakage $\epsilon_p^{(k)}$. Similarly, when $T$ remains constant, increasing the gradient distortion $\Delta^{(k)}$ similarly leads to decreased privacy leakage $\epsilon_p^{(k)}$, as stronger perturbations make attacks less effective. These monotonicity properties demonstrate that larger datasets, stricter privacy budgets, and higher gradient distortions all contribute to reducing privacy leakage in federated learning systems. For detailed proofs, refer to \pref{sec:analysis_upper_bound}.}

In the previous section, we provided an upper bound for attack complexity, where we analyzed factors such as privacy leakage and gradient distortion to determine the maximum impact that an attacker can impose. However, from a practical perspective, the lower bound of attack complexity may be of greater interest. In the following, we further analyze the lower bound of attack complexity, derived by combining the definitions of privacy leakage and gradient distortion. This lower bound analysis offers a more comprehensive understanding of attacker behavior.

\textcolor{black}{\begin{thm}[Lower Bound of Attack Complexity with Gradient Expectation Bound]\label{lower-bound}
Let us define the variables as follows: \( p \), \( c_a \), \( c_b \), and \( c_2 \) are constants; \( D^{(k)} \) represents the local dataset of client \( k \); \( \tau \) denotes the error threshold; \( T \) is the number of attacker iterations; \( m \) is the dimension of the model gradient vector; and \( \epsilon \) denotes the level of MBP enforced by the protection mechanism. Combining the definitions of privacy leakage and gradient distortion, we provide the lower bound for \( T \) with at least \( 1 - \gamma^{(k)} \) confidence as follows:
\begin{align}
    T^{1 - p} \ge \frac{c_2 \cdot c_b}{
        4[\tau \cdot (1-\epsilon_p^{(k)}) + \sqrt{\frac{\ln(2/\gamma^{(k)})}{2|D^{(k)}|}}]
        - c_a \cdot \Delta^{(k)}
    }
\end{align}
By invoking the protection complexity result from Section~\ref{sec:protection_bounds}, where Theorems~\ref{thm:lower-bound-eps-DP} and~\ref{thm: upper_bound_utility_loss} jointly establish a tight distortion bound $\mathbb{E}[\|\distpara - \truepara\|_2^2] = \Theta \left( \frac{m}{\min\{\epsilon^2, \epsilon\}} \right)$, we estimate the gradient distortion as $\Delta^{(k)} = \sqrt{ \frac{m}{\min\{\epsilon, \epsilon^2\}} }$, which reflects the expected magnitude of perturbation introduced by the protection mechanism. Substituting this into the inequality yields:
\begin{align}
    T^{1-p} \ge \frac{c_2 \cdot c_b}{
        4[\tau \cdot (1-\epsilon_p^{(k)}) + \sqrt{\frac{\ln(2/\gamma^{(k)})}{2|D^{(k)}|}}]
        - c_a \cdot \left( \sqrt{\frac{m}{\min\{\epsilon^2, \epsilon\}}+ \sqrt{\frac{m\cdot ln(1/\gamma^{(k)})}{2}}} \right)
    }
\end{align}
\end{thm}}

Specifically, Theorem~\ref{lower-bound} provides a lower bound for the number of iterations \( T \) that an attacker must perform. The inequality shows that this lower bound is determined by three key factors: the attack threshold \( \tau \), the model dimension \( m \), and the privacy budget \( \epsilon \). A larger threshold \( \tau \) indicates a more relaxed accuracy requirement for successful reconstruction, thus reducing the required number of iterations. In contrast, a higher model dimension \( m \) implies that the attacker must recover more parameters, which increases the difficulty of the attack and raises the lower bound on \( T \). Similarly, a smaller privacy budget \( \epsilon \) (meaning stronger MBP) increases the distortion in the shared gradients and further increases the number of iterations needed to breach the protection.

These findings demonstrate the interdependence among \( \epsilon_p^{(k)} \), the model dimension \( m \), the privacy budget \( \epsilon \), and the number of attack iterations \( T \). They provide critical theoretical insights into how different factors influence attack feasibility, highlighting the importance of balancing privacy budgets with protection strength in the design of privacy-preserving systems. The constants $c_a$, $c_b$, $c_2$, and $p$ in the attack complexity bounds characterize the bi-Lipschitz properties of gradients and the efficiency of the attacker's optimization. While their precise values are difficult to determine analytically for real-world systems, they are crucial for the theoretical framework. Qualitatively, the scaling relationships derived (e.g., stronger privacy increases attack complexity) remain robust. However, the numerical magnitude of required attack iterations ($T$) is highly sensitive to these constants, emphasizing the need for empirical validation across diverse system scales and data distributions to bridge theoretical insights with practical FL system design.

\section{Conclusion and Discussion}\label{sec:conclusion}
In this paper, we have presented a novel theoretical framework for understanding and quantifying the intricate interplay between attack and protection complexities in privacy-preserving federated learning. Recognizing the growing threat of gradient inversion attacks and the inherent trade-offs in privacy mechanisms, our work introduces formal definitions for "Attack Complexity" and "Protection Complexity," providing a rigorous basis for analyzing the security and privacy landscape of FL.

Our contributions include the precise definition of Attack Complexity, which measures the minimum resources an adversary needs to reconstruct private data within a given error threshold. Complementarily, we defined Protection Complexity as the expected distortion introduced by privacy-preserving mechanisms. A cornerstone of our theoretical analysis is the derivation of tight bounds for Protection Complexity based on MBP, demonstrating its scaling with model dimensionality and the privacy budget. This approach offers a distinct advantage over analyses solely reliant on LDP by providing a more general characterization of information distortion from an adversary’s belief update perspective. Furthermore, we established comprehensive upper and lower bounds for Attack Complexity, explicitly linking it to critical factors such as privacy leakage, gradient distortion (quantified by MBP), model dimension, and the number of attacker iterations.

The key insight gleaned from our analysis is the quantitative elucidation of the fundamental trade-offs in federated learning. We have shown that stronger privacy guarantees, achieved through increased gradient distortion (i.e., a smaller MBP $\epsilon$), directly translate into higher attack complexity, demanding significantly more resources from an adversary to achieve a successful data reconstruction. Conversely, relaxing privacy requirements reduces protection complexity but makes the system more vulnerable to attacks. This framework provides a principled way to understand and balance the costs associated with both attacking and defending private data in FL.

Our findings offer critical guidance for the design and deployment of more secure and efficient federated learning systems. By understanding the precise relationships between privacy budgets, model characteristics, and adversarial effort, developers can make informed decisions about the level of privacy protection to implement, optimizing for both utility and security. For instance, a Branch and Bound Algorithm could leverage our findings to optimize the design of privacy-preserving FL systems. By defining objective functions that incorporate both the utility of the FL system (related to Protection Complexity $C$) and the security against attacks (related to Attack Complexity $S_k(\alpha)$), the algorithm could systematically explore the vast design space of privacy parameters (like MBP $\epsilon$), model architectures, and protection mechanisms. At each node of the search tree, representing a specific FL configuration, the algorithm could compute the expected Attack Complexity (e.g., the minimum iterations $T$ required for an adversary to achieve a certain reconstruction quality $\alpha$, as per Theorem \ref{upper-bound} and \ref{lower-bound}) and the associated Protection Complexity (e.g., the expected distortion $C$, as per Theorem \ref{thm:lower-bound-eps-DP} and \ref{thm: upper_bound_utility_loss}). If, for a given configuration, the calculated attack cost (e.g., $T$) exceeds a predefined threshold for feasibility, or if the benefit of a successful attack (e.g., the value of reconstructed data) is outweighed by the resources required, that branch of the search space could be effectively pruned. This systematic pruning, guided by our derived complexity bounds, would enable efficient identification of FL system configurations that offer an optimal balance between privacy, utility, and resilience against gradient inversion attacks, moving beyond heuristic approaches to a more principled design methodology.

Future work will focus on several promising directions. Firstly, we plan to empirically validate our theoretical bounds using various real-world datasets and machine learning models to confirm their practical implications. Secondly, we intend to extend our framework to encompass other types of privacy attacks and protection mechanisms, such as membership inference attacks and more advanced cryptographic techniques. Thirdly, exploring the impact of different distance metrics and optimization algorithms on attack complexity would provide further nuanced insights. Finally, we aim to investigate adaptive privacy mechanisms that dynamically adjust protection levels based on real-time assessments of attack and protection complexities, paving the way for more resilient and resource-efficient federated learning ecosystems.

% \section{ACKNOWLEDGMENTS}
% This work was supported by the National Science and Technology Major Project under Grant 2022ZD0115301.

\bibliography{main}
\bibliographystyle{ACM-Reference-Format}

\newpage

\onecolumn
\appendix

\section{Analysis for bounds for distortion extent using Bayesian privacy.}\label{sec: lower_bound_for_federated_utility_loss}

\subsection{Analysis for Theorem \ref{thm:lower-bound-eps-DP}}\label{sec:analysis_lower_bound_pc}
\begin{proof}
To facilitate the subsequent proof, we first present the necessary assumption and technical lemma.

\begin{assumption}{Compositional differential privacy bounds.}\label{assump:eps-DP-app}
For each $i$ and $t$, and for all $W^{\cal{D}}_{i,t}$, the channel mapping $X_i$ to $W^{\cal{D}}_{i,t}$ is $\epsilon_{i,t}(W^{\cal{D}}_{i,t})$-approximately differentially private. There exists a constant $\varepsilon_{\mathrm{kl}}$ such that
\begin{equation}
\sum_{i=1}^n \sum_{t=1}^T \mathbb{E} \left[
\min\left\{
\frac{\left(\epsilon_{i,t}^{2}(W^{\cal{D}}_{i,t})\right)}{\log 2},
\epsilon_{i,t}(W^{\cal{D}}_{i,t})
\right\}
\right]
\leq n \cdot \varepsilon_{\mathrm{kl}},
\end{equation}
where $\varepsilon_{\mathrm{kl}}$ represents the average per-user KL-privacy leakage over all users and time steps. The expectation is taken over the randomness in the privatized outputs $W^{\cal{D}}_{i,t}$.
\end{assumption}

\begin{lem}[Assouad’s Method]\label{lem:assouad-app}
Let the conditions of the previous paragraph hold and let $\mathcal{V}$ induce a $\delta$-separation in Hamming metric. Here, $\delta$-separation in Hamming metric ensures that any two candidate vectors differ in at least $\delta$ coordinates, enforcing sufficient separation for effective lower bounding. Then
\[
\mathfrak{M}_n(\truepara(\mathcal{P}), L, Q) \geq \frac{\delta}{2} \sum_{j=1}^m \left(1 - \left\| \mathbb{P}_{+j} - \mathbb{P}_{-j} \right\|_{\mathrm{TV}} \right),
\]
where $\mathbb{P}_{+j}$ and $\mathbb{P}_{-j}$ denote the output distributions corresponding to $V_j = +1$ and $V_j = -1$, respectively.
\end{lem}

We now construct Bernoulli-type distributions to apply Lemma~\ref{lem:assouad-app}. Let $\mathbb{P}_{-1} = \text{Bernoulli}(1/2)$ and, for some $\delta < 1$, define $\mathbb{P}_{1} = \text{Bernoulli}((1+\delta)/2)$. Then
\[
\left| \log \frac{d\mathbb{P}_1}{d\mathbb{P}_{-1}} \right| \leq -\log(1 - \delta).
\]

Consequently, for $V$ uniformly distributed over $\{-1,1\}^m$, we obtain the mutual information bound:
\[
I(V; \distpara) \leq 2 \left( \frac{1}{1 - \delta} - 1 \right)^2 I(X; \distpara) = \frac{2 \delta^2}{1 - 2\delta + \delta^2} I(X; \distpara).
\]

Also, the binary divergence satisfies: $\beta(\mathbb{P}_{-1}, \mathbb{P}_1) \leq \frac{2 \delta^2}{(1 - \delta)^2}$.
Therefore, for any $\delta < 1$, through corollary 14 of~\cite{duchi2019lower}, we have:
\[
\sum_{j=1}^m \mathbb{P}(\hat{V}_j(\distpara) \neq V_j) \geq \frac{m}{2} \left( 1 - \sqrt{ \frac{7(2 - \delta)}{1 - \delta} \cdot \frac{2 \delta^2}{(1 - \delta)^2} \cdot \frac{n \cdot \varepsilon_{\mathrm{kl}}}{m} } \right).
\]

Now take $\delta^2 = c \cdot \min\left\{1, \frac{m}{n \cdot \min(\epsilon, \epsilon^2)} \right\}$, and for all privacy types assumed in the theorem we have $\varepsilon_{\mathrm{kl}} \asymp \min\{ \epsilon, \epsilon^2 \}$. The separation in Lemma~\ref{lem:assouad-app} is at least $\delta/2$, and we conclude by absorbing constants:
\[
\mathbb{E}[\|\distpara - \truepara\|_2^2] \geq c \cdot \max\left( \frac{m}{\min(\epsilon, \epsilon^2)}, 1 \right),
\]
as claimed. By Lemma~\ref{Relationship between MBP and LDP}, we know that $\epsilon$-LDP implies $\epsilon$-MBP. Therefore, the theorem follows.

\end{proof}

\subsection{Analysis for Theorem \ref{thm: upper_bound_utility_loss}}\label{sec:analysis_upper_bound_pc}

\textcolor{black}{\begin{proof} 
Assume that the transformation from $\truepara$ to $\distpara$ satisfies $\epsilon$-MBP.  
By Lemma~\ref{Relationship between MBP and LDP}, it follows that this mechanism also satisfies $\epsilon$-LDP.  
Therefore, by the proof of Proposition 4 of \cite{bhowmick2018protection}, we conclude that:
\[
\mathbb{E}[\|\distpara - \truepara\|_2^2] \leq c \cdot \left( \frac{m}{\epsilon} \vee \frac{m}{(e^\epsilon - 1)^2}\right),
\]
for some constant $c > 0$.
For $\epsilon \leq 1$, we use the fact that $(e^\epsilon - 1)^2 \approx \epsilon^2$, which yields:
\[
\mathbb{E}[\|\distpara - \truepara\|_2^2] = \mathcal{O}\left( \frac{m}{\epsilon^2} \right).
\]
For $\epsilon \geq 1$, since $(e^\epsilon - 1)^2 \gtrsim \epsilon^2$, we obtain:
\[
\mathbb{E}[\|\distpara - \truepara\|_2^2] = \mathcal{O}\left( \frac{m}{\epsilon} \right).
\]
Moreover, the matching lower bound is provided by Theorem~\ref{thm:lower-bound-eps-DP}, which implies:
\[
\mathbb{E}[\|\distpara - \truepara\|_2^2] = \Omega\left( \frac{m}{\min\{\epsilon^2, \epsilon\}} \right).
\]
Combining both bounds, we obtain the tight characterization:
\[
\mathbb{E}[\|\distpara - \truepara\|_2^2] = \Theta\left( \frac{m}{\min\{\epsilon^2, \epsilon\}} \right),
\]
as claimed.
\end{proof}}

\section{Analysis for Theorem \ref{upper-bound}}\label{sec:analysis_upper_bound}

%\textbf{Assumption 1: Bi-Lipschitz Condition of Gradients}

\begin{assumption}[Bi-Lipschitz Condition of Gradients]\label{assump: two-sided Lipschitz}
We assume that the gradients satisfy the bi-Lipschitz condition:
\begin{equation}\label{eq: bi_lipschitz}
c_a \|\nabla L(\theta, X_1, Y) - \nabla L(\theta, X_2, Y)\| \le \|X_1 - X_2\| \le c_b \|\nabla L(\theta, X_1, Y) - \nabla L(\theta, X_2, Y)\|,
\end{equation}
\end{assumption}

%\textbf{Assumption 2: Self-bounded Regret}
\begin{assumption}[Self-bounded Regret]\label{assump: bounds_for_optimization_alg}
Let \( T \) represent the number of learning iterations by the attacker. We assume that the error values satisfy
\begin{equation}
c_0 \cdot T^{\frac{1}{2}} \le \sum_{t = 1}^T \|\nabla L(\theta, X_t, Y_t) - \nabla L(\theta, \tilde{X}, Y_t)\| \triangleq \Theta(T^{\frac{1}{2}}) \le c_2 \cdot T^{\frac{1}{2}},
\end{equation}
where \( c_0 \) and \( c_2 \) are constants, \( X_t \) represents the reconstructed data by the attacker in the \( t \)-th iteration, and \( \tilde{X} \) is the dataset used to generate the perturbed gradients.
\end{assumption}

Let \( X_{t,i}^{(k)} \) represent the \( i \)-th data reconstructed by the attacker for client \( k \) in the \( t \)-th iteration, and \( X_{i}^{(k)} \) represent the \( i \)-th data of client \( k \). Let \( D_{attacker}^{(k)} = \{(\tilde{X}_1^{(k)}, Y_1^{(k)}), \ldots, (\tilde{X}_{m_k}^{(k)}, Y_{m_k}^{(k)})\} \) represent the local data used to generate the distorted gradient \( \tilde{g} \). We can derive:

\begin{align*}
    &\frac{1}{|D_{attacker}^{(k)}|} \sum_{i = 1}^{|D_{attacker}^{(k)}|} \frac{1}{T} \sum_{t = 1}^T \|X_{t,i}^{(k)} - X_{i}^{(k)}\| \nonumber\\
    &\ge \frac{1}{|D_{attacker}^{(k)}|} \sum_{i = 1}^{|D_{attacker}^{(k)}|} \frac{1}{T} \sum_{t = 1}^T \|\tilde{X}_{i}^{(k)} - X_{i}^{(k)}\| - \frac{1}{|D_{attacker}^{(k)}|} \sum_{i = 1}^{|D_{attacker}^{(k)}|} \frac{1}{T} \sum_{t = 1}^T \|X_{t,i}^{(k)} - \tilde{X}_{i}^{(k)}\| \nonumber\\
    &\ge c_a \cdot \frac{1}{T} \sum_{t = 1}^T \frac{1}{|D_{attacker}^{(k)}|} \sum_{i = 1}^{|D_{attacker}^{(k)}|} \|\nabla L(\theta, \tilde{X}_{i}^{(k)}, Y_i^{(k)}) - \nabla L(\theta, X_{i}^{(k)}, Y_i^{(k)})\| \nonumber\\
    &\quad - c_b \cdot \frac{1}{T} \sum_{t = 1}^T \frac{1}{|D_{attacker}^{(k)}|} \sum_{i = 1}^{|D_{attacker}^{(k)}|} \|\nabla L(\theta, X_{t,i}^{(k)}, Y_i^{(k)}) - \nabla L(\theta, \tilde{X}_{i}^{(k)}, Y_i^{(k)})\| \nonumber\\
    &\ge c_a \cdot \frac{1}{T} \sum_{t = 1}^T \left\| \frac{1}{|D_{attacker}^{(k)}|} \sum_{i = 1}^{|D_{attacker}^{(k)}|} \left( \nabla L(\theta, \tilde{X}_{i}^{(k)}, Y_i^{(k)}) - \nabla L(\theta, X_{i}^{(k)}, Y_i^{(k)}) \right) \right\| \nonumber\\
    &\quad - c_b \cdot \frac{1}{T} \sum_{t = 1}^T \frac{1}{|D_{attacker}^{(k)}|} \sum_{i = 1}^{|D_{attacker}^{(k)}|} \|\nabla L(\theta, X_{t,i}^{(k)}, Y_i^{(k)}) - \nabla L(\theta, \tilde{X}_{i}^{(k)}, Y_i^{(k)})\| \nonumber\\
    & = c_a \cdot \Delta^{(k)} - c_b \cdot \frac{1}{T} \sum_{t = 1}^T \frac{1}{|D_{attacker}^{(k)}|} \sum_{i = 1}^{|D_{attacker}^{(k)}|} \|\nabla L(\theta, X_{t,i}^{(k)}, Y_i^{(k)}) - \nabla L(\theta, \tilde{X}_{i}^{(k)}, Y_i^{(k)})\| \label{eq: data_gap_terms},
\end{align*}

where the second inequality follows from Assumption \ref{assump: two-sided Lipschitz}, the third inequality follows from the triangle inequality (\(\|a\| + \|b\| \ge \|a + b\|\)), and the final inequality follows from the definition of gradient distortion.

Furthermore, from Assumption \ref{assump: bounds_for_optimization_alg}, we have:

\begin{align*}
    \sum_{t = 1}^T \|\nabla L(\theta, X_{t,i}^{(k)}, Y_i^{(k)}) - \nabla L(\theta, \tilde{X}_{i}^{(k)}, Y_i^{(k)})\| \le c_2 \cdot T^{1/2}.
\end{align*}

Thus, we further obtain:

\begin{align*}
    &c_a \cdot \Delta^{(k)} - c_b \cdot \frac{1}{T} \sum_{t = 1}^T \frac{1}{|D_{attacker}^{(k)}|} \sum_{i = 1}^{|D_{attacker}^{(k)}|} \|\nabla L(\theta, X_{t,i}^{(k)}, Y_i^{(k)}) - \nabla L(\theta, \tilde{X}_{i}^{(k)}, Y_i^{(k)})\| \\
    &\ge c_a \Delta^{(k)} - c_2 \cdot c_b T^{-\frac{1}{2}} \\
\end{align*}

By Hoeffding's inequality, with at least \( 1 - \delta \) probability, we have:

\begin{align}
    \left|\frac{1}{t} \sum_{i = 1}^{t} X_i - \E\left[\frac{1}{t} \sum_{i = 1}^{t} X_i\right]\right|\le \sqrt{\frac{\ln(2/\delta)}{2t}},
\end{align}
where $X_i\in [0,1]$.

Thus, with at least \( 1 - \gamma^{(k)} \) probability, we obtain:

\begin{align*}
   &\frac{1}{|D_{attacker}^{(k)}|} \sum_{i = 1}^{|D_{attacker}^{(k)}|} \frac{1}{T} \sum_{t = 1}^T \frac{\|X_{t,i}^{(k)} - X_{i}^{(k)}\|}{\tau} \\
   &\le \mathbb{E} \left[ \frac{1}{|D_{attacker}^{(k)}|} \sum_{i = 1}^{|D_{attacker}^{(k)}|} \frac{1}{T} \sum_{t = 1}^T \frac{\|X_{t,i}^{(k)} - X_{i}^{(k)}\|}{\tau} \right] + \sqrt{\frac{\ln(2/\gamma^{(k)})}{2|D_{attacker}^{(k)}|}} \\
\end{align*}

Note that we set
\begin{align}
    |D_{attacker}^{(k)}| = |D^{(k)}|.
\end{align}
Thus, we have
\begin{align*}
   &\frac{1}{|D_{attacker}^{(k)}|} \sum_{i = 1}^{|D_{attacker}^{(k)}|} \frac{1}{T} \sum_{t = 1}^T \frac{\|X_{t,i}^{(k)} - X_{i}^{(k)}\|}{\tau} \\
   &\le \mathbb{E} \left[ \frac{1}{|D_{attacker}^{(k)}|} \sum_{i = 1}^{|D_{attacker}^{(k)}|} \frac{1}{T} \sum_{t = 1}^T \frac{\|X_{t,i}^{(k)} - X_{i}^{(k)}\|}{\tau} \right] + \sqrt{\frac{\ln(2/\gamma^{(k)})}{2|D_{attacker}^{(k)}|}} \\
   & = 1 - \epsilon_p^{(k)} + \sqrt{\frac{\ln(2/\gamma^{(k)})}{2|D^{(k)}|}} \\
\end{align*}
where $|D_{attacker}^{(k)}|$  represents the size of the mini-batch (the training set of the attacker).

Furthermore, we obtain:

\begin{align*}
     1 - \epsilon_p^{(k)} + \sqrt{\frac{\ln(2/\gamma^{(k)})}{2|D^{(k)}|}}
     & \ge \frac{1}{|D^{(k)}|} \sum_{i = 1}^{|D^{(k)}|} \frac{1}{T} \sum_{t = 1}^T \frac{\|X_{t,i}^{(k)} - X_{i}^{(k)}\|}{\tau} \\
     &\ge \frac{c_a \Delta^{(k)} - c_2 \cdot c_b T^{-\frac{1}{2}}}{\tau}.
\end{align*}

Further simplifying, we obtain:

\begin{align*}
T^{\frac{1}{2}} \leq \frac{c_2 \cdot c_b}{{c_a} \cdot \Delta^{(k)} - \tau \cdot \left[1 - \epsilon_p^{(k)} + \sqrt{\frac{\ln(2/\gamma^{(k)})}{2|D^{(k)}|}} \right]}
\end{align*}

\textcolor{black}{Let the perturbation vector be $N=\distpara-\truepara \in R^m$, and $N=(n_1,n_2,\dots,n_m) $. We assume that the components of $N$ are bounded within [0,1], and are independent and identically distributed random variables. Our goal is to establish a lower bound on $\Delta^{(k)} = \| \distpara - \truepara \|_2 = \|N\|_2$ that holds with high probability.
We begin with the result from Theorem \ref{thm: upper_bound_utility_loss}:
\[
\mathbb{E}[\|\distpara - \truepara\|_2^2] = \Theta\left( \frac{m}{\min\{\epsilon^2, \epsilon\}} \right).
\]
So we have that $\mathbb{E}[\|N\|_2^2] = \Theta\left( \frac{m}{\min\{\epsilon^2, \epsilon\}} \right)$. Now, consider the sequence of i.i.d. random variables $Z_i = n_i^2$ for $i=1,\dots,m$. The sample mean of this sequence is $\frac{1}{m}\sum\limits_{i=1}^mZ_i=\frac{1}{m}\|N\|_2^2$. We can apply Hoeffding's inequality to bound the deviation of this sample mean:
\begin{align*}
    Pr\left(\mathbb{E}\left[\frac{1}{m}\sum_{i=0}^m Z_i\right]-\frac{1}{m}\sum\limits_{i=1}^mZ_i \ge t\right) \le \exp(-2mt^2),
\end{align*}
that also mean:
\begin{align*}
    Pr\left(\mathbb{E}\left[\frac{1}{m}\|N\|_2^2\right]-\frac{1}{m}\|N\|_2^2 \ge t\right) \le \exp(-2mt^2).
\end{align*}
This is equivalent to:
\begin{align*}
    Pr\left(\frac{1}{m}\|N\|_2^2\ge\mathbb{E}\left[\frac{1}{m}\|N\|_2^2\right] - t\right) \ge 1 -  \exp(-2mt^2).
\end{align*}
Substituting $\exp(-2mt^2) = \gamma^{(k)}$, with probability at least$1-\gamma^{(k)}$, we have:
\begin{align*}
    \|N\|_2^2\ge\mathbb{E}\left[\|N\|_2^2\right] - \sqrt{\frac{m\cdot ln(1/\gamma^{(k)})}{2}}
\end{align*}
Now we can express the lower bound for the gradient distortion $\Delta^{(k)} = \|N\|_2$:
\begin{align*}
    \Delta^{(k)}= \| \distpara - \truepara \|_2\ge\sqrt{\frac{m}{\min\{\epsilon^2, \epsilon\}}- \sqrt{\frac{m\cdot ln(1/\gamma^{(k)})}{2}}}
\end{align*}}

Substitute the existing results:
\begin{align*}
    T^{\frac{1}{2}} \leq \frac{c_2 \cdot c_b}{c_a \cdot \left( \sqrt{\frac{m}{\min\{\epsilon^2, \epsilon\}}- \sqrt{\frac{m\cdot ln(1/\gamma^{(k)})}{2}}} \right) - \tau \cdot \left[1 - \epsilon_p^{(k)} + \sqrt{\frac{\ln(2/\gamma^{(k)})}{2|D^{(k)}|}} \right]}
\end{align*}

\section{Analysis for Theorem \ref{lower-bound}}\label{sec: anlaysis_for_lower_bound}

\begin{assumption}[Bi-Lipschitz Condition]\label{assump: two_sided_Lipschitz}
   For any two datasets $d_1$ and $d_2$, assume that $c_a ||g(d_1) - g(d_2)||\le ||d_1 - d_2||\le c_b ||g(d_1) - g(d_2)||$.
\end{assumption}

\begin{assumption}[Self-bounded Regret]\label{assump: regret_bound}
Let \( T \) represent the number of learning iterations by the attacker. We assume that the error values satisfy
\begin{equation}
c_0 \cdot T^{p} \le \sum_{t = 1}^T \|\nabla L(\theta, X_t, Y_t) - \nabla L(\theta, \tilde{X}, Y_t)\| \triangleq \Theta(T^{p}) \le c_2 \cdot T^{p},
\end{equation}
where \( c_0 \) and \( c_2 \) are constants, \( X_t \) represents the reconstructed data by the attacker in the \( t \)-th iteration, and \( \tilde{X} \) is the dataset used to generate the perturbed gradients.
\end{assumption}

Assume that the semi-honest attacker uses an optimization algorithm to infer the original dataset of client $k$ based on the released parameter $\tilde{\theta}^{(k)}$. Let \( D_{attacker}^{(k)} = \{(\tilde{X}_1^{(k)}, Y_1^{(k)}), \ldots, (\tilde{X}_{m_k}^{(k)}, Y_{m_k}^{(k)})\} \) represent the dataset satisfying that  $g(D_{attacker}^{(k)}) = \tilde{\theta}^{(k)}$, and $\theta^{(k)}$ represent the dataset satisfying that $g( D^{(k)}) = \theta^{(k)}$, where $\theta^{(k)}$ represents the original parameter, $\tilde{\theta}^{(k)}$ represents the protected parameter. Let \( X_{t,i}^{(k)} \) represent the \( i \)-th data reconstructed by the attacker for client \( k \) in the \( t \)-th iteration, and \( X_{i}^{(k)} \) represent the \( i \)-th data of client \( k \). Let $\Delta^{(k)} = ||\frac{1}{|D^{(k)}|}\sum_{m = 1}^{|D^{(k)}|} \nabla L(\theta, X_{t,i}^{(k)}, Y_i^{(k)}) - \frac{1}{|D^{(k)}|}\sum_{m = 1}^{|D^{(k)}|}\nabla L(\theta, \tilde{X}_{i}^{(k)}, Y_i^{(k)}) ||$ represent the distortion of the parameter. The expected regret of the optimization algorithm in a total of $T$ ($ T > 0$) rounds is $\Theta(T^p)$.
If $\Delta^{(k)}\ge\frac{2c_2 c_b}{c_a}\cdot T^{p-1}$, then we have that
\begin{align}
        T^{1-p} \ge \frac{c_2 \cdot c_b}{
        4[\tau \cdot (1-\epsilon_p^{(k)}) + \sqrt{\frac{\ln(2/\gamma^{(k)})}{2|D^{(k)}|}}]
        - c_a \cdot \left( \sqrt{ \frac{m}{\min\{\epsilon, \epsilon^2\}} }
        + \sqrt{\frac{\ln(2/\gamma^{(k)})}{2|D^{(k)}|}} \right)
    }
\end{align}
where $c_a$ and $c_b$ are introduced in \pref{assump: two_sided_Lipschitz}, and $c_2$ is introduced in \pref{assump: regret_bound}.

\begin{proof}
Recall the privacy leakage of the client is defined as

\begin{equation}
    \epsilon_p^{(k)} = 1 - \mathbb{E}\left[\frac{1}{|D^{(k)}|}\sum_{i=1}^{|D^{(k)}|}\frac{1}{T}\sum_{t=1}^T \frac{\|\tilde{X}_{t,i}^{(k)} - X_{i}^{(k)}\|}{\tau}\right],
\end{equation}
To protect privacy, client $k$ selects a protection mechanism $M^{(k)}$, which maps the original parameter $\theta^{(k)}$ to a protected parameter $\tilde{\theta}^{(k)}$. After observing the protected parameter, a semi-honest adversary infers the private information using the optimization approaches. Let $X_{t}^{(k)}$ represent the reconstructed data at iteration $t$ using the optimization algorithm. Therefore the cumulative regret over $T$ rounds
\begin{align*}
    R(T) & = \sum_{t = 1}^{T} [||\nabla L(\theta, X_t, Y_t) - \tilde{\theta}|| - ||\nabla L(\theta, \tilde{X}, Y_t) - \tilde{\theta}||]\\
    & = \sum_{t = 1}^{T} [||\nabla L(\theta, X_t, Y_t) - \nabla L(\theta, \tilde{X}, Y_t)||]\\
    & = \Theta(T^p).
\end{align*}
Therefore, we have
\begin{align}\label{eq: regret_bounds}
   c_0\cdot T^p \le \sum_{t = 1}^{T}||\nabla L(\theta, X_t, Y_t) - \nabla L(\theta, \tilde{X}, Y_t)|| = \Theta(T^p) \le c_2\cdot T^p,
\end{align}
where $c_0$ and $c_2$ are constants independent of $T$.

Let $x$ and $\wtilde x$ represent two data. From our assumption, we have that
\begin{align}
    c_a ||g(x) - g(\wtilde x)||\le ||x - \wtilde x||\le c_b ||g(x) - g(\wtilde x)||.
\end{align}

Let $X_{t,i}^{(k)}$ represent the reconstructed $i$-th data at iteration $t$ using the optimization algorithm. We have that
\begin{align*}
    &\frac{1}{|D_{attacker}^{(k)}|} \sum_{i = 1}^{|D_{attacker}^{(k)}|} \frac{1}{T} \sum_{t = 1}^T \|X_{t,i}^{(k)} - X_{i}^{(k)}\| \nonumber\\
    &\ge \frac{1}{|D_{attacker}^{(k)}|} \sum_{i = 1}^{|D_{attacker}^{(k)}|} \frac{1}{T} \sum_{t = 1}^T \|\tilde{X}_{i}^{(k)} - X_{i}^{(k)}\| - \frac{1}{|D_{attacker}^{(k)}|} \sum_{i = 1}^{|D_{attacker}^{(k)}|} \frac{1}{T} \sum_{t = 1}^T \|X_{t,i}^{(k)} - \tilde{X}_{i}^{(k)}\| \nonumber\\
    &\ge c_a \cdot \frac{1}{T} \sum_{t = 1}^T \frac{1}{|D_{attacker}^{(k)}|} \sum_{i = 1}^{|D_{attacker}^{(k)}|} \|\nabla L(\theta, \tilde{X}_{i}^{(k)}, Y_i^{(k)}) - \nabla L(\theta, X_{i}^{(k)}, Y_i^{(k)})\| \nonumber\\
    &\quad - c_b \cdot \frac{1}{T} \sum_{t = 1}^T \frac{1}{|D_{attacker}^{(k)}|} \sum_{i = 1}^{|D_{attacker}^{(k)}|} \|\nabla L(\theta, X_{t,i}^{(k)}, Y_i^{(k)}) - \nabla L(\theta, \tilde{X}_{i}^{(k)}, Y_i^{(k)})\| \nonumber\\
    &\ge c_a \cdot \frac{1}{T} \sum_{t = 1}^T \left\| \frac{1}{|D_{attacker}^{(k)}|} \sum_{i = 1}^{|D_{attacker}^{(k)}|} \left( \nabla L(\theta, \tilde{X}_{i}^{(k)}, Y_i^{(k)}) - \nabla L(\theta, X_{i}^{(k)}, Y_i^{(k)}) \right) \right\| \nonumber\\
    &\quad - c_b \cdot \frac{1}{T} \sum_{t = 1}^T \frac{1}{|D_{attacker}^{(k)}|} \sum_{i = 1}^{|D_{attacker}^{(k)}|} \|\nabla L(\theta, X_{t,i}^{(k)}, Y_i^{(k)}) - \nabla L(\theta, \tilde{X}_{i}^{(k)}, Y_i^{(k)})\| \nonumber\\
\end{align*}

where the second inequality is due to
\begin{align*}
&\frac{1}{|D_{attacker}^{(k)}|} \sum_{i = 1}^{|D_{attacker}^{(k)}|} \frac{1}{T} \sum_{t = 1}^T \|\tilde{X}_{i}^{(k)} - X_{i}^{(k)}\| \nonumber\\
&\ge c_a\cdot\frac{1}{T} \sum_{t = 1}^T \frac{1}{|D_{attacker}^{(k)}|} \sum_{i = 1}^{|D_{attacker}^{(k)}|} \|\nabla L(\theta, \tilde{X}_{i}^{(k)}, Y_i^{(k)}) - \nabla L(\theta, X_{i}^{(k)}, Y_i^{(k)})\| \nonumber\\
\end{align*}
and
\begin{align*}
&\frac{1}{|D_{attacker}^{(k)}|} \sum_{i = 1}^{|D_{attacker}^{(k)}|} \frac{1}{T} \sum_{t = 1}^T \|X_{t,i}^{(k)} - \tilde{X}_{i}^{(k)}\| \nonumber\\
&\le c_b \cdot \frac{1}{T} \sum_{t = 1}^T \frac{1}{|D_{attacker}^{(k)}|} \sum_{i = 1}^{|D_{attacker}^{(k)}|} \|\nabla L(\theta, X_{t,i}^{(k)}, Y_i^{(k)}) - \nabla L(\theta, \tilde{X}_{i}^{(k)}, Y_i^{(k)})\| \nonumber\\
\end{align*}

From the definition of distortion extent, we know that
\begin{align}
   \Delta^{(k)} = ||\frac{1}{|D^{(k)}|}\sum_{i = 1}^{|D^{(k)}|} \nabla L(\theta, X_{t,i}^{(k)}, Y_i^{(k)}) - \frac{1}{|D^{(k)}|}\sum_{i = 1}^{|D^{(k)}|}\nabla L(\theta, \tilde{X}_{i}^{(k)}, Y_i^{(k)}) ||.
\end{align}

We set
\begin{align}
    |D_{attacker}^{(k)}| = |D^{(k)}|.
\end{align}

By Hoeffding's inequality, with at least \( 1 - \delta \) probability, we have:

\begin{align}
    \left|\frac{1}{t} \sum_{i = 1}^{t} X_i - \E\left[\frac{1}{t} \sum_{i = 1}^{t} X_i\right]\right|\le \sqrt{\frac{\ln(2/\delta)}{2t}},
\end{align}
where $X_i\in [0,1]$.

Thus, with at least \( 1 - \gamma^{(k)} \) probability, we obtain:

\begin{align*}
   &\frac{1}{|D_{attacker}^{(k)}|} \sum_{i = 1}^{|D_{attacker}^{(k)}|} \frac{1}{T} \sum_{t = 1}^T \frac{\|X_{t,i}^{(k)} - X_{i}^{(k)}\|}{\tau} \\
   &\le \mathbb{E} \left[ \frac{1}{|D_{attacker}^{(k)}|} \sum_{i = 1}^{|D_{attacker}^{(k)}|} \frac{1}{T} \sum_{t = 1}^T \frac{\|X_{t,i}^{(k)} - X_{i}^{(k)}\|}{\tau} \right] + \sqrt{\frac{\ln(2/\gamma^{(k)})}{2|D_{attacker}^{(k)}|}} \\
\end{align*}

Note that we set
\begin{align}
    |D_{attacker}^{(k)}| = |D^{(k)}|.
\end{align}
Thus, we have
\begin{align*}
   &\frac{1}{|D_{attacker}^{(k)}|} \sum_{i = 1}^{|D_{attacker}^{(k)}|} \frac{1}{T} \sum_{t = 1}^T \frac{\|X_{t,i}^{(k)} - X_{i}^{(k)}\|}{\tau} \\
   &\le \mathbb{E} \left[ \frac{1}{|D_{attacker}^{(k)}|} \sum_{i = 1}^{|D_{attacker}^{(k)}|} \frac{1}{T} \sum_{t = 1}^T \frac{\|X_{t,i}^{(k)} - X_{i}^{(k)}\|}{\tau} \right] + \sqrt{\frac{\ln(2/\gamma^{(k)})}{2|D_{attacker}^{(k)}|}} \\
   & = 1 - \epsilon_p^{(k)} + \sqrt{\frac{\ln(2/\gamma^{(k)})}{2|D^{(k)}|}} \\
\end{align*}
where $|D_{attacker}^{(k)}|$  represents the size of the mini-batch (the training set of the attacker).

Therefore, we have
\begin{align*}
    & \tau \cdot (1-\epsilon_p^{(k)}) + \sqrt{\frac{\ln(2/\gamma^{(k)})}{2|D^{(k)}|}} \\
    & \ge \tau \cdot (1-\epsilon_p^{(k)}) = \mathbb{E}\left[\frac{1}{|D^{(k)}|}\sum_{i=1}^{|D^{(k)}|}\frac{1}{T}\sum_{t=1}^T \|\tilde{X}_{t,i}^{(k)} - X_{i}^{(k)}\|\right]\\
    &\ge  c_a\Delta^{(k)} - c_b\cdot\frac{1}{T}\sum_{t = 1}^{T}\sum_{i = 1}^{|D^{(k)}|}||\nabla L(\theta, X_t, Y_t) - \nabla L(\theta, \tilde{X}, Y_t)||\\
    &\ge c_a\Delta^{(k)} - c_2\cdot c_b T^{p-1}\\
    &\ge\frac{1}{2}\max\{c_a\Delta^{(k)}, c_2\cdot c_b T^{p-1}\} \\
    &\ge \frac{c_a\Delta^{(k)} + c_2\cdot c_b T^{p-1}}{4}.
\end{align*}

Therefore, we have that
\begin{align}
    T^{1-p}\ge\frac{c_2\cdot c_b}{4[\tau \cdot (1-\epsilon_p^{(k)}) + \sqrt{\frac{\ln(2/\gamma^{(k)})}{2|D^{(k)}|}}] -c_a\Delta^{(k)}}
\end{align}

\textcolor{black}{Let the perturbation vector be $N=\distpara-\truepara \in R^m$, and $N=(n_1,n_2,\dots,n_m) $. We assume that the components of $N$ are bounded within [0,1], and are independent and identically distributed random variables. Our goal is to establish a lower bound on $\Delta^{(k)} = \| \distpara - \truepara \|_2 = \|N\|_2$ that holds with high probability.
We begin with the result from Theorem \ref{thm: upper_bound_utility_loss}:
\[
\mathbb{E}[\|\distpara - \truepara\|_2^2] = \Theta\left( \frac{m}{\min\{\epsilon^2, \epsilon\}} \right).
\]
So we have that $\mathbb{E}[\|N\|_2^2] = \Theta\left( \frac{m}{\min\{\epsilon^2, \epsilon\}} \right)$. Now, consider the sequence of i.i.d. random variables $Z_i = n_i^2$ for $i=1,\dots,m$. The sample mean of this sequence is $\frac{1}{m}\sum\limits_{i=1}^mZ_i=\frac{1}{m}\|N\|_2^2$. We can apply Hoeffding's inequality to bound the deviation of this sample mean:
\begin{align*}
    P\left(\mathbb{E}\left[\frac{1}{m}\sum_{i=0}^m Z_i\right]-\frac{1}{m}\sum\limits_{i=1}^mZ_i \ge t\right) \le \exp(-2mt^2),
\end{align*}
that also mean:
\begin{align*}
    P\left(\mathbb{E}\left[\frac{1}{m}\|N\|_2^2\right]-\frac{1}{m}\|N\|_2^2 \ge t\right) \le \exp(-2mt^2).
\end{align*}
This is equivalent to:
\begin{align*}
    P\left(\frac{1}{m}\|N\|_2^2\ge\mathbb{E}\left[\frac{1}{m}\|N\|_2^2\right] - t\right) \ge 1 -  \exp(-2mt^2).
\end{align*}
Substituting $\exp(-2mt^2) = \gamma^{(k)}$, with probability at least$1-\gamma^{(k)}$, we have:
\begin{align*}
    \|N\|_2^2\le\mathbb{E}\left[\|N\|_2^2\right] + \sqrt{\frac{m\cdot ln(1/\gamma^{(k)})}{2}}
\end{align*}
Now we can express the lower bound for the gradient distortion $\Delta^{(k)} = \|N\|_2$:
\begin{align*}
    \Delta^{(k)}=\| \distpara - \truepara \|_2\le\sqrt{\frac{m}{\min\{\epsilon^2, \epsilon\}}+ \sqrt{\frac{m\cdot ln(1/\gamma^{(k)})}{2}}}
\end{align*}}

Substitute the existing results:
\begin{align}
    T^{1-p} \ge \frac{c_2 \cdot c_b}{
        4[\tau \cdot (1-\epsilon_p^{(k)}) + \sqrt{\frac{\ln(2/\gamma^{(k)})}{2|D^{(k)}|}}]
        - c_a \cdot \left( \sqrt{\frac{m}{\min\{\epsilon^2, \epsilon\}}+ \sqrt{\frac{m\cdot ln(1/\gamma^{(k)})}{2}}} \right)
    }
\end{align}
\end{proof}

% add mbp estimation
% begin
{
\color{black}
\section{Practical Estimation for Maximum Bayesian Privacy}\label{sec: anlaysis_for_estimation_mbp}

\begin{assumption}
[Maximality of Leakage at True Parameter]
     \label{assump: maximality_of_leakage_at_true_parameter} 
We assume that the amount of Bayesian privacy leaked from the true parameter is maximum overall possible parameters. That is, for any $w \in \mathcal{W}$, we have that
\begin{align}
    \frac{f_{D|W}(d|w^*)}{f_D(d)} \ge \frac{f_{D|W}(d|w)}{f_D(d)},
\end{align}
where $w^*$ represents the true model parameter.
\end{assumption}

Recall that the $\epsilon-$Maximum Bayesian Privacy(MBP) guarantee, as defined in Definition \ref{def:mbp}, is given by:
\begin{align*}
    e^{-\epsilon} \leq \frac{f_{D|W}(d \mid w)}{f_D(d)} \leq e^{\epsilon}.
\end{align*}

This is equivalent to stating that for any possible released information $w \in \mathcal{W}$ and  any private data $d \in \mathcal{D}$, the following holds:
\begin{align*}
    \left| \log\left(\frac{f_{D|W}(d|w)}{f_D(d)}\right) \right| \le \epsilon .
\end{align*}

The MBP parameter $\epsilon$ (also denoted as $\xi$ in Definition 5.2) is therefore the tightest such bound:
\begin{align*}
    \epsilon = \max_{w \in \mathcal{W}, d \in \mathcal{D}} \left| \log\left(\frac{f_{D|W}(d|w)}{f_D(d)}\right) \right| .
\end{align*}

From the Assumption \ref{assump: maximality_of_leakage_at_true_parameter}, we have that:
\begin{align}
    \epsilon = \max_{d \in \mathcal{D}} \left| \log\left(\frac{f_{D|W}(d|w^*)}{f_D(d)}\right) \right| .
\end{align}
Consequently, the task of estimating the MBP parameter $\epsilon$ reduces to the more specific problem of estimating the posterior probability density $f_{D|W}(d|w^)$ given the worst-case release. Based on its estimate, $\hat f_{D|W}(d|w^*)$, the estimated privacy parameter $\hat\epsilon$ is given by:
\begin{align}
    \hat\epsilon = \max_{d \in \mathcal{D}} \left| \log\left(\frac{\hat f_{D|W}(d|w^*)}{f_D(d)}\right) \right| .
\end{align}
\subsection{Estimating the value of $\hat f_{D|W}(d|w^*)$}\label{sec:Estimating the value of f_{D|W}(d|w^*)}

The calculation of the MBP parameter $\epsilon$ hinges on estimating the conditional probability $f_{D|W}(d|w^*)$. To this end, we adopt the attack simulation methodology described in Section 7.3 of \cite{zhang2024meta}. 

The core idea of this algorithm is to estimate the conditional probability by generating perturbed data and comparing it with the original data during the parameter update process. By performing multiple iterations and accumulating counts, the algorithm provides an estimation of the conditional probability.This algorithm takes a mini-batch of data $d$, with batch size S, and a set of parameters $\{w_m\}_{m=1}^M$ as input. 

The key steps of the algorithm are as follows:
 
\begin{itemize}
    \item \textbf{Compute Gradients and Update Parameters:} Using mini-batch samples, calculate the gradient of the loss function $\mathcal{L}$ with respect to each parameter $w_m$, and update the parameters using gradient descent.

    \item \textbf{Simulated Data Reconstruction:} For each parameter $w_m$, perform $T_{sim}$ iterations to generate perturbed data by using the DLG method.

    \item \textbf{Success Criterion:} Determine if the original data is successfully recovered by checking if the difference between the original data $z_d$ and the perturbed data is less than a predefined threshold $\Omega$.

    \item \textbf{Estimate Conditional Probability:} Estimate the conditional probability $\hat{f}_{D|W}(d|w_m)$ by dividing the count of successfully recovered instances by the total number of attempts (iterations $T$ multiplied by batch size $S$).
\end{itemize}

This process yields the estimated conditional probability $\hat{f}_{D|W}(d|w_m)$, a key component for calculating the MBP parameter. As analyzed in Appendix \ref{sec:Estimation Error of MBP}, the estimation error between the estimated value and the true value diminishes as the number of trials $T_{sim}$ increases, thereby ensuring the reliability of our MBP evaluation.

\subsection{Analysis for Theorem \ref{thm: the_estimation_error_of_epsilon}}\label{sec:Estimation Error of MBP}
In this section, We leverage Lemma  \ref{lemma: reliability_of_simulated_reconstruction} to derive Theorem \ref{thm: the_estimation_error_of_epsilon}.
\begin{lem}
[Reliability of Simulated Reconstruction, adapted from \cite{zhang2024meta}] \label{lemma: reliability_of_simulated_reconstruction} 
With probability at least $1-2\exp\left(-\frac{\beta^2T_{sim}\kappa_1(d)}{3}\right)$, we have that
\begin{align}
    |\hat{\kappa}_1(d) - \kappa_1(d)| \leq \beta \kappa_1(d),
\end{align}
where $\kappa_1(d) = f_{D|W}(d|w^*)$, $\beta$ represent error tolerance or precision, $T_{sim}$ represent the total number of simulation rounds.
\end{lem}
\begin{proof}

Using the inequality for the estimation error of $\kappa_1(d)$ provided by Lemma \ref{lemma: reliability_of_simulated_reconstruction}, we can derive its relationship with $T_{sim}$

Lemma \ref{lemma: reliability_of_simulated_reconstruction} actually tells us that for a specific data point $d$, we have that:
\begin{align*}
    Pr \left[|\hat{\kappa}_1(d) - \kappa_1(d)| \leq \beta \kappa_1(d)\right] \ge 1-2\exp\left(-\frac{\beta^2T_{sim}\kappa_1(d)}{3}\right)
\end{align*}

To obtain an explicit convergence rate, we set the confidence level to $1-\delta$, where $\delta = 2\exp\left(-\frac{\beta^2T_{sim}\kappa_1(d)}{3}\right)$. By inverting this equation, we can solve for $\beta$:
\begin{align}
    \beta = \sqrt{\frac{3\ln(2/\delta)}{T_{sim}\kappa_1(d)}}.
\end{align}

Plugging this expression for $\beta$ into the inequality $|\hat{\kappa}1(d) - \kappa_1(d)| \le \beta \kappa_1(d)$, we obtain a bound on the error related to $T_{sim}$ with at least $1-\delta$ probability:
\begin{align*}
    &|\hat{\kappa}_1(d) - \kappa_1(d)| \le \sqrt{\frac{3\ln(2/\delta)}{T_{sim}\kappa_1(d)}} \kappa_1(d)
    \nonumber\\
    & =  \sqrt{\frac{3\ln(2/\delta) \kappa_1(d)}{T_{sim}}}.
\end{align*}

For a fixed dataset size $|\mathcal{D}|$, confidence level $\delta$, the terms on the right-hand side, apart from $T_{sim}$ and $\kappa_1(d)$, are constants, we conclude that the error of $\hat{\kappa}_1(d)$ as:

\begin{equation}\label{eq: error_of_kappa}
    |\hat{\kappa}_1(d) - \kappa_1(d)| = O\left(\sqrt{\frac{\ln(2/\delta)\kappa_1(d)}{T_{sim}}}\right) .
\end{equation}

Next, we leverage the estimation error of $\hat{\kappa}_1(d)$ to analyze the estimation error of $\hat{\epsilon}$, which is $|\hat{\epsilon} - \epsilon|$. By definition, we have:
\begin{align*}
    &|\hat\epsilon-\epsilon| = \left| \max_{d \in \mathcal{D}} \left| \log\left(\frac{\hat \kappa_1(d)}{f_D(d)}\right) \right| - \max_{d \in \mathcal{D}} \left| \log\left(\frac{\kappa_1(d)}{f_D(d)}\right) \right| \right|.
\end{align*}

Using the inequality $| \max\limits_i a_i- \max\limits_i b_i | \le \max\limits_i|a_i-b_i|$, we can further bound this expression:
\begin{align*}
    &|\hat\epsilon-\epsilon| \le \left| \max_{d \in \mathcal{D}}\left| \left|\log\left(\frac{\hat \kappa_1(d)}{f_D(d)}\right)\right| - \left|\log\left(\frac{\kappa_1(d)}{f_D(d)}\right)\right| \right| \right|\nonumber\\
    &\le \left| \max_{d \in \mathcal{D}}\left| \log\left(\frac{\hat \kappa_1(d)}{f_D(d)}\right) - \log\left(\frac{\kappa_1(d)}{f_D(d)}\right) \right| \right|
    \nonumber\\
    & =\max\limits_{d \in \mathcal{D}} \left| \log(\hat{\kappa}_1(d)) - \log(\kappa_1(d)) \right|.
\end{align*}

The second inequality holds due to the triangle inequality $| |a| - |b| | \le |a-b|$.

By the Mean Value Theorem, for each $d$, there exists a value $\gamma(d)$ between $\hat{\kappa}_1(d)$ and $\kappa_1(d)$ such that:
\begin{align*}
\left| \log(\hat{\kappa}_1(d)) - \log(\kappa_1(d)) \right| = \left| \frac{1}{\gamma(d)} (\hat{\kappa}_1(d) - \kappa_1(d)) \right|.
\end{align*}

Referencing \pref{eq: error_of_kappa}, as the number of trials $T_{sim} \to \infty$, the estimator $\hat{\kappa}_1(d)$ converges in probability to its true value $\kappa_1(d)$. Since $\gamma(d)$ lies between $\hat{\kappa}_1(d)$ and $\kappa_1(d)$, it follows from the Squeeze Theorem that $\gamma(d) \to \kappa_1(d)$.

This convergence implies that for a sufficiently large $T_{sim}$, the term $1/\gamma(d)$ is asymptotically equivalent to $1/\kappa_1(d)$. Consequently, we can express $1/\gamma(d)$ using Big-O notation as $O(1/\kappa_1(d))$:
\begin{align*}
    \left| \log(\hat{\kappa}_1(d)) - \log(\kappa_1(d)) \right| = O\left(\left|\frac{\hat{\kappa}_1(d) - \kappa_1(d)}{\kappa_1(d)}\right|\right).
\end{align*}

Since we have shown that $|\hat{\kappa}_1(d) - \kappa_1(d)| = O(\sqrt{\frac{\ln(2/\delta)\kappa_1(d)}{T_{sim}}})$. With at least $1 -\delta$ probability, it follows that:

\begin{align*}
\max\limits_{d \in \mathcal{D}} \left| \log(\hat{\kappa}_1(d)) - \log(\kappa_1(d)) \right| = O\left(\max\limits_{d \in \mathcal{D}} \sqrt{\frac{\ln(2/\delta)}{T_{sim}\cdot\kappa_1(d)}}\right) = O\left(\sqrt{\frac{\ln(2/\delta)}{T_{sim}\cdot\min\limits_{d \in \mathcal{D}}\kappa_1(d)}}\right).
\end{align*}

For a specific dataset $\mathcal{D}$ and model, $\min\limits_{d\in\mathcal{D}}\kappa_1(d)$ is regarded as a constant. We arrive at the final error bound for the estimation of $\hat{\epsilon}$ with at least $1 -\delta$ probability :

\begin{align}
    &|\hat\epsilon-\epsilon| = O\left(\sqrt{\frac{\ln(2/\delta)}{T_{sim}}}\right)
\end{align}

From the definition of Big O, we have:
\begin{align*}
    \hat\epsilon-\zeta\le\epsilon\le\hat\epsilon+\zeta, 
\end{align*}
where $\zeta = c\cdot\left(\sqrt{\frac{\ln(2/\delta)}{T_{sim}}}\right)$.

We now analyze the implications of this estimation error on the main theoretical results of our paper. Recall that Theorem \ref{thm: upper_bound_utility_loss} provides a tight bound for Protection Complexity based on the true MBP parameter $\epsilon$:
\begin{align*}
    \E[\|\distpara - \truepara\|_2^2] = O \left(\frac{m}{\min\{\epsilon^2, \epsilon\}}\right),
\end{align*}
and
\begin{align*}
    \E[\|\distpara - \truepara\|_2^2] = \Omega \left(\frac{m}{\min\{\epsilon^2, \epsilon\}}\right).
\end{align*}

Since the protection complexity is a monotonically decreasing function of $\epsilon$, we can bound it using the extremities of this confidence interval. The lower bound on protection complexity is achieved when $\epsilon$ is at its maximum $\hat\epsilon+\zeta$, and the upper bound is achieved when $\epsilon$ is at its minimum $\hat\epsilon-\zeta$.

Substituting this relationship into the bound from Theorem \ref{thm: upper_bound_utility_loss}, we obtain a modified characterization of the protection complexity:

\begin{align}
    \E[\|\distpara - \truepara\|_2^2] = \Omega \left(\frac{m}{\min\{\left(\hat\epsilon\ + \zeta\right)^2, \hat\epsilon+ \zeta\}}\right),
\end{align}
and
\begin{align}
    \E[\|\distpara - \truepara\|_2^2] = O \left(\frac{m}{\min\{\left(\hat\epsilon- \zeta\right)^2, \hat\epsilon- \zeta\}}\right).
\end{align}
\end{proof}
}

\end{document}